\documentclass[sigconf, nonacm, pdfa]{acmart}

\newcommand\vldbdoi{10.14778/3675034.3675044}
\newcommand\vldbpages{2528 - 2540}
\newcommand\vldbvolume{17}
\newcommand\vldbissue{10}
\newcommand\vldbyear{2024}
\newcommand\vldbauthors{\authors}
\newcommand\vldbtitle{\shorttitle} 
\newcommand\vldbavailabilityurl{https://github.com/YimingQiao/Blitzcrank}
\newcommand\vldbpagestyle{empty} 

\usepackage[ruled,vlined,linesnumbered]{algorithm2e} % For algorithms
\usepackage{dsfont}
\usepackage{multicol}
\usepackage{multirow}
\usepackage{enumitem}
\usepackage{xcolor}
\usepackage[noabbrev,capitalize,nameinlink]{cleveref}
\usepackage{marginnote}
\usepackage[a-2b]{pdfx}

\DeclareMathOperator{\ShiftLeft}{{\LARGE \texttt{<<}}}
\DeclareMathOperator{\ShiftRight}{{\LARGE \texttt{>>}}}

\newtheorem{thm}{Theorem}

\newcommand{\work}{\textsc{Blitzcrank}\xspace}

\newcommand{\revision}[1]{\textcolor{black}{#1}}

\begin{document}
\title{Blitzcrank: Fast Semantic Compression for In-memory Online Transaction Processing}

%%
%% The "author" command and its associated commands are used to define the authors and their affiliations.
\author{Yiming Qiao}
\affiliation{
  \institution{Tsinghua University}
}
\email{qiaoym21@mails.tsinghua.edu.cn}

\author{Yihan Gao}
\affiliation{
  % \institution{Independent Researcher}
}
\email{gaoyihan@gmail.com}

\author{Huanchen Zhang}
\authornote{Huanchen Zhang is also affiliated with Shanghai Qi Zhi Institute.}
\affiliation{
  \institution{Tsinghua University}
}
\email{huanchen@tsinghua.edu.cn}

\begin{abstract}
We present \work, a high-speed semantic compressor designed for OLTP databases.
Previous solutions are inadequate for compressing row-stores:
they suffer from either low compression factor due to a coarse compression granularity
or suboptimal performance due to the inefficiency in handling dynamic data sets.
To solve these problems, we first propose novel semantic models that support fast inferences
and dynamic value set for both discrete and continuous data types.
We then introduce a new entropy encoding algorithm, called delayed coding,
that achieves significant improvement in the decoding speed compared to modern arithmetic
coding implementations.
We evaluate \work in both standalone microbenchmarks
and a multicore in-memory row-store using the TPC-C benchmark.
Our results show that \work achieves a sub-microsecond latency for decompressing a random tuple
while obtaining high compression factors.
%This leads to an 85\% memory reduction in our end-to-end TPC-C evaluation
This leads to an 85\% memory reduction in the TPC-C evaluation
with a moderate (19\%) throughput degradation.
For data sets larger than the available physical memory,
\work help the database sustain a high throughput for more transactions
before the I/O overhead dominates.
\end{abstract}

\maketitle

\pagestyle{\vldbpagestyle}
\begingroup\small\noindent\raggedright\textbf{PVLDB Reference Format:}\\
\vldbauthors. \vldbtitle. PVLDB, \vldbvolume(\vldbissue): \vldbpages, \vldbyear.\\
\href{https://doi.org/\vldbdoi}{doi:\vldbdoi}
\endgroup
\begingroup
\renewcommand\thefootnote{}\footnote{\noindent
This work is licensed under the Creative Commons BY-NC-ND 4.0 International License. Visit \url{https://creativecommons.org/licenses/by-nc-nd/4.0/} to view a copy of this license. For any use beyond those covered by this license, obtain permission by emailing \href{mailto:info@vldb.org}{info@vldb.org}. Copyright is held by the owner/author(s). Publication rights licensed to the VLDB Endowment. \\
\raggedright Proceedings of the VLDB Endowment, Vol. \vldbvolume, No. \vldbissue\ %
ISSN 2150-8097. \\
\href{https://doi.org/\vldbdoi}{doi:\vldbdoi} \\
}\addtocounter{footnote}{-1}\endgroup

\ifdefempty{\vldbavailabilityurl}{}{
\vspace{.3cm}
\begingroup\small\noindent\raggedright\textbf{PVLDB Artifact Availability:}\\
The source code, data, and/or other artifacts have been made available at \url{\vldbavailabilityurl}.
\endgroup
}

\section{Introduction} 
In-memory database management systems (DBMSs) offer low latency and high throughput for online transaction processing (OLTP) workloads when the working set fits in memory \cite{van2018managing, sudhir2021replicated, tu2013speedy}. 
For data sets beyond physical memory, their performance degrades quickly because of expensive random I/Os to fetch tuples.
Although DRAM price has been decreasing, memory is still a limiting resource because of the increasing price gap between DRAM and SSDs~\cite{ziegler2022scalestore, haas2020exploiting}.

Applying compression can increase the capacity of in-memory DBMSs with the same hardware cost, thus reducing or even eliminating disk accesses to improve performance~\cite{DBLP:conf/sigmod/ZhangAPKMS16, 10.1145/3318464.3380583Hope}.
However, most existing compression techniques are designed for column-stores: they target read-mostly workloads with large batched processing \cite{kuschewski2023btrblocks, abadi2013design, abadi2006integrating, boncz2020fsst, damme2017lightweight, DBLP:conf/sigmod/FoufoulasSSI21, polychroniou2015efficient, barbarioli2023hierarchical, lasch2020faster}.
To compress in-memory row-stores efficiently, the compression schemes must satisfy additional requirements.
First, random access to tuples must be fast (e.g., sub-microsecond) because OLTP applications demand low query latencies.
For example, Amazon found that a 100 ms increase in latency would lead to a 1\% drop in sales \cite{DBLP:conf/wosp/HegerHMO17, DBLP:journals/pvldb/LerschSOL20}.
Second, the compression algorithms must handle newly inserted/updated tuples efficiently because OLTP workloads are typically write-heavy~\cite{oltp_workload_ibm}. A frequent reconstruction of the compression model is usually unacceptable because it brings too much performance overhead~\cite{DBLP:conf/vldb/PossP03CompressionOracle, DBLP:conf/vldb/RamanS06}.
Unfortunately, existing compression schemes are inadequate when serving OLTP workloads: they suffer from either low compression factor (i.e. the original size divided by the compressed size) due to a coarse compression granularity or suboptimal performance due to the inefficiency in handling dynamic data sets.

\setlength{\belowcaptionskip}{-4pt}
\begin{figure}[t!]
  \centering
  \includegraphics[width=0.8\linewidth]{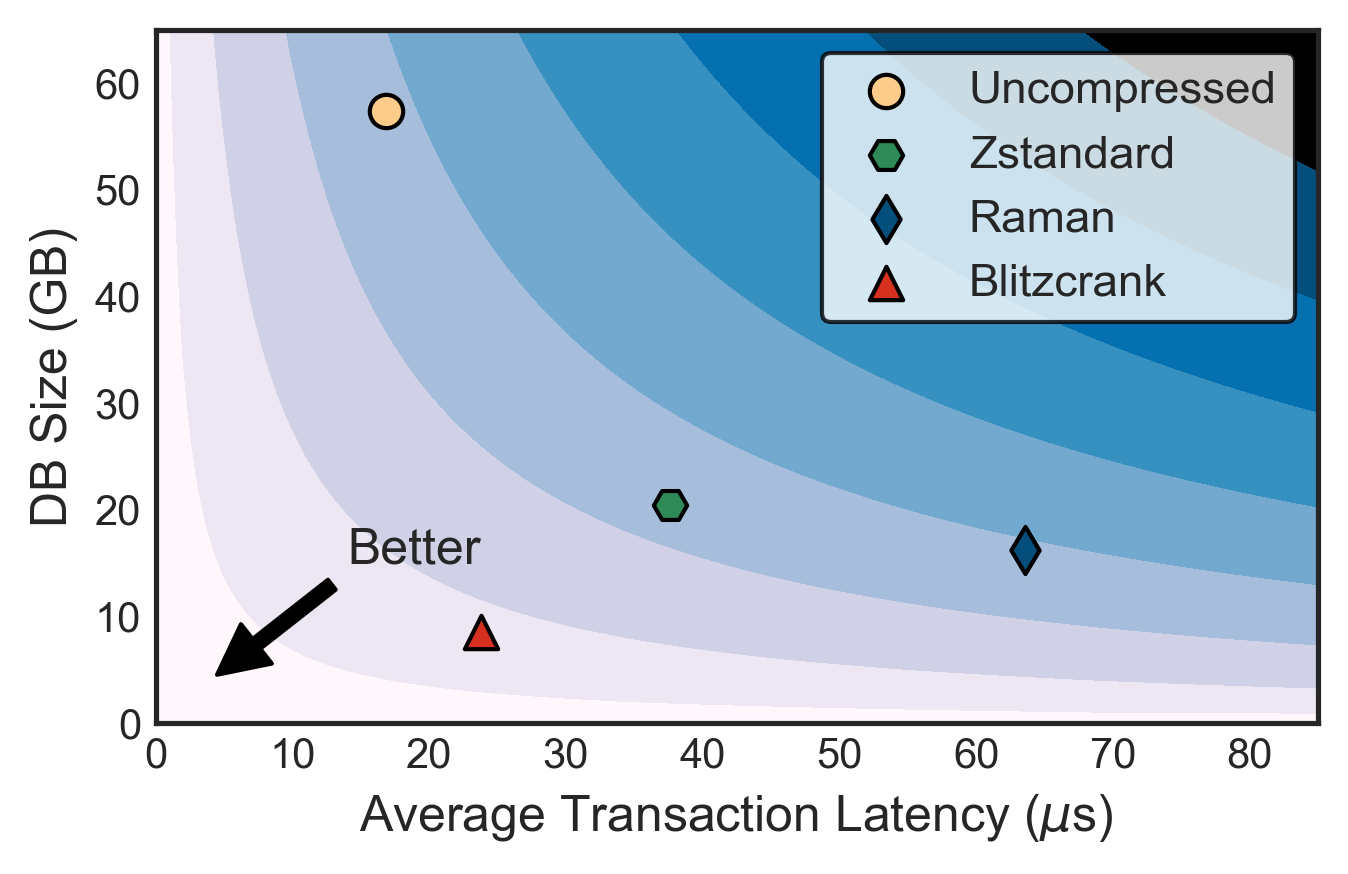}
  \vspace{-2ex}
  \caption{DB Size vs. Latency \textnormal{- \work makes the size-latency trade-offs more attractive compared to other tools in TPC-C.}}
  \label{fig:tpcc_hotmap}
  \vspace{-1.5ex}
\end{figure}
\textbf{Coarse Retrieval Granularity.} Modern general-purpose block compression algorithms such as Zstandard have high decompression throughput (up to 500 MB/s). They are widely used in operating systems, databases, file systems, and computer networks~\cite{zstd-application, yang2018elastic, DBLP:conf/mm/0002JZM021Multimedia}.
However, to access a single tuple, they must decompress the entire compression block \cite{DBLP:journals/tit/ZivL77gzip}, causing \textit{high random-access latency}.
One solution is to compress each tuple individually, but the compression factor suffers
because the algorithms prefer a longer context to create an effective dictionary in the sliding window \cite{DBLP:journals/tit/ZivL77gzip, DBLP:journals/tit/ZivL78}.

\textbf{Inefficient Handling of New Tuples.} The classic Raman's approach \cite{DBLP:conf/vldb/RamanS06} concatenates Huffman-coded values into variable-length tuples and then reorders the rows and columns so that it achieves a better compression factor using delta encoding. Although this solution compresses row-oriented data well, it cannot compress unseen values unless it initiates an expensive model reconstruction because such a solution relies on static dictionaries.

The above approaches are considered ``syntactic'' because they treat the uncompressed data simply as consecutive bytes~\cite{MacKay2003InformationTheory}. 
\revision{Semantic compression, on the other hand, leverages the high-level semantics in a relational table, such as value distributions and functional dependencies between columns to achieve better compression \cite{DBLP:conf/vldb/RamanS06}. Unlike the above syntactic methods that rely on static dictionaries, the semantic approach can use the same probability models to compress new tuples effectively as long as the attribute values follow the modeled distributions \cite{DBLP:conf/sigmod/BabuGR01}.}
Existing semantic compression methods, however, provide limited support for different data types, and their model inferences are slow.
For example, Squish and the more recent DeepSqueeze take 324 and 127 seconds, respectively, to compress a 75 MB relational table~\cite{DBLP:conf/sigmod/IlkhechiCGMFSC20, DBLP:conf/kdd/GaoP16}.

% \marginpar[R1.W3]{}\textcolor{blue}{Semantic compression, on the other hand, leverages the high-level semantics in a relational table, such as column data distributions and functional dependencies between columns to achieve better compression \cite{DBLP:conf/vldb/RamanS06}. This approach addresses above compression challenges by focusing on the underlying data patterns rather than mere byte sequences. Unlike traditional methods relying on byte repetition dictionaries, this approach has the potential to compress row-stores, even for new tuples, as long as attribute values maintain their distribution \cite{DBLP:conf/sigmod/BabuGR01}.}
% Existing semantic compression methods, however, provide limited support for different data types, and their model inferences are slow.
% For example, Squish and the more recent DeepSqueeze take 324 and 127 seconds, respectively, to compress a 75 MB relational table~\cite{DBLP:conf/sigmod/IlkhechiCGMFSC20, DBLP:conf/kdd/GaoP16}.

In this paper, we show that semantic compression can be fast, and it has potential beyond large archive compression. We present \textbf{\work}, a high-speed semantic compressor designed for OLTP databases. \work improves compression through both data modeling and data encoding \cite{DBLP:journals/tcom/ClearyW84}.
\vspace{-0.5em}
\begin{itemize}[leftmargin=*]
    %\item For data modeling, \work introduces novel semantic models for common types in databases and supports newly tuples compression. Meanwhile, it speeds up model inference substantially. It takes \work less than one second to compress the aforementioned 75 MB table.
    \item For data modeling, \work introduces novel semantic models that allow fast encoding/decoding for both discrete- and continuous-value columns. It takes \work less than one second to compress the aforementioned 75 MB table.

    %\item For data encoding, we propose a novel fine-grained entropy coding algorithm, called \textit{delayed coding}. It offers near-entropy compression as with Arithmetic Coding \cite{langdon1984introduction} while achieving a faster decompression speed compared to the modern asymmetric number system (ANS) \cite{DBLP:journals/corr/abs-2106-06438-ans}.
    \item For data encoding, we propose \textit{delayed coding}, a novel fine-grained encoding algorithm that offers near-entropy compression as with Arithmetic Coding \cite{langdon1984introduction} while achieving a faster decompression speed compared to the modern asymmetric number system (ANS) \cite{DBLP:journals/corr/abs-2106-06438-ans}.
\end{itemize}

We first compared \work against state-of-the-art compressors that apply to \textit{row-stores},
including Zstandard \cite{DBLP:journals/rfc/rfc8878zstd} and Raman's approach \cite{DBLP:conf/vldb/RamanS06}, in standalone microbenchmarks based on real data sets.
\work achieves a sub-microsecond latency (fastest among the baselines) for decompressing a random tuple while obtaining high compression factors. 
We then integrated \work, along with baseline compressors, into the in-memory OLTP database, Silo \cite{tu2013speedy} and evaluated it using the TPC-C benchmark~\cite{tpcc}. 
\revision{The results are summarized in \cref{fig:tpcc_hotmap}. Compared to uncompressed tables, \work reduces memory usage by 85\% with a moderate (19\%) throughput degradation.}
% \marginpar[R1.W1.1, R5.Minor.1]{R1.W1.1, R5.Minor.1}\textcolor{blue}{\cref{fig:tpcc_hotmap} shows that \work reduces memory usage by 85\% with a moderate (19\%) throughput degradation.}
Compared to using Zstandard, \work achieves a $2.4 \times$ higher compression factor and is 76\% faster.
When the data set exceeds the physical memory limit, \work greatly helps the database sustain high throughput and execute $4 \times$ more transactions within the same amount of time.

The paper makes the following contributions.
First, we identify the inefficiency of existing compression algorithms for OLTP databases from both data modeling and data encoding perspectives.
Second, we introduce novel semantic models designed for fast inferences for discrete and continuous data types.
Third, we propose the new delayed coding that is significantly faster than variants of arithmetic coding while achieving near-entropy compression.
Finally, we build \work based on the above technologies and show that semantic compression can be fast enough to make trade-offs between performance and space much more attractive than previous solutions when integrated into an in-memory OLTP database, such as Silo~\cite{tu2013speedy}.

\section{Preliminaries}

This section provides the necessary background information to understand the design of \work. \cref{sec: arithmetic coding} describes the classic arithmetic coding, which is the basis of our proposed delayed coding. \cref{sec:structure-learning} introduces the existing structure learning techniques adopted in \work to leverage functional dependencies between columns for compression.

% \textcolor{blue}{\marginpar[]{R1.W1.2}Semantic compression has two parts: entropy coding, which encodes data based on element frequency, and structure learning, which uncovers the data patterns within a table. They work together to improve compression efficiency.
% }

\subsection{Arithmetic Coding }
\label{sec: arithmetic coding}
\begin{figure}[t!]
  \centering
  \includegraphics[width=\linewidth]{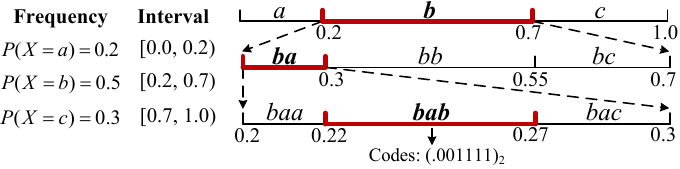}
  \vspace{-5ex}
  \caption{An Example of Arithmetic Coding \textnormal{- Arithmetic coding maps each possible string to disjoint probability intervals.}}
  \vspace{-1ex}
  \label{fig:arithmetic coding}
\end{figure}

\label{sec:structure-learning}
\begin{figure}[t!]
  \centering
  \includegraphics[width=\linewidth]{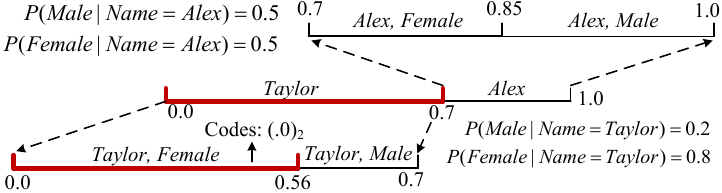}
  \vspace{-5ex}
  \caption{An Example of Column Correlation \textnormal{- Probabilities of column ``gender'' depends on the ``name'' column value.}}
  \label{fig:correlation}
  \vspace{-1ex}
\end{figure}

\begin{figure*}[t!]
  \centering
  \includegraphics[width=0.95\linewidth]{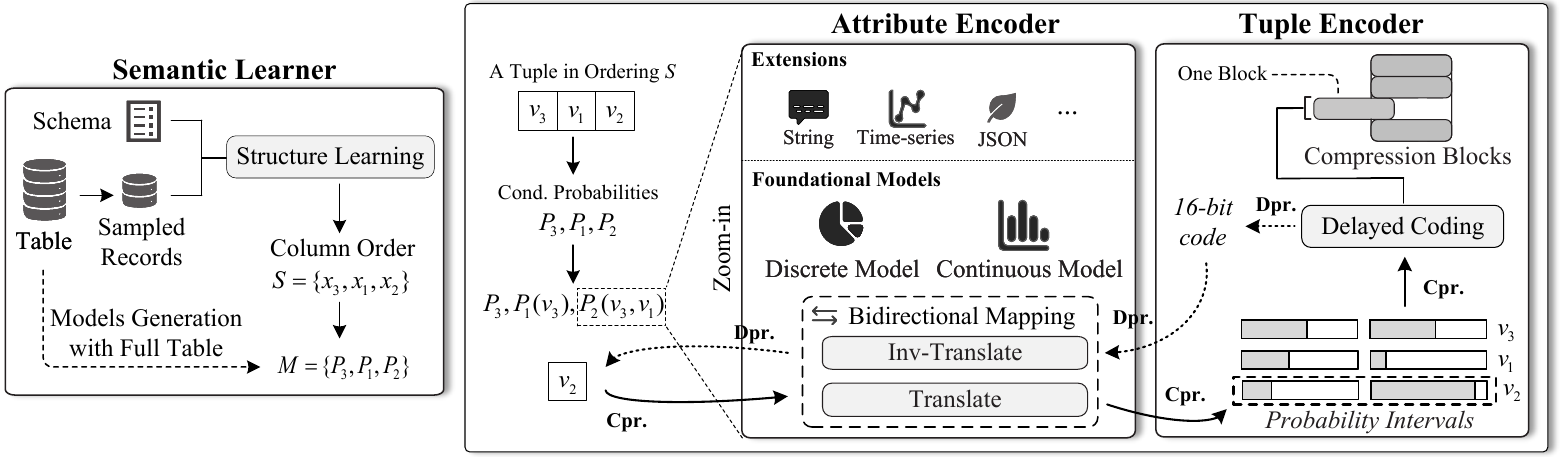}
  \vspace{-2ex}
  \caption{\revision{\work - \textnormal{Semantic Learner (SL), Attribute Encoder (AE), and Tuple Encoder (TE) are three components of \work. }}}
  \label{fig:overview}
  \vspace{-1.6ex}
\end{figure*}

Arithmetic coding is one of the most widely used entropy codings for lossless compression \cite{langdon1984introduction, DBLP:journals/tcom/ClearyW84}.
Unlike Huffman coding \cite{moffat2019huffman} that encodes symbols individually, arithmetic coding compresses the entire message into a single fraction $0 \leq q < 1$ with arbitrary precision. Compared to Huffman coding, arithmetic coding can achieve a higher compression factor. Arithmetic coding represents the current information as an interval, defined by two numbers (initially $[0,1)$). Each encoding step in arithmetic coding divides the current interval into smaller sub-intervals according to the probability distribution of the alphabet and selects the one that represents the next symbol to be encoded. For example, as shown in \cref{fig:arithmetic coding}, the probability distribution of alphabet $\{a, b, c\}$ is 0.2, 0.5, and 0.3, respectively. To encode a message ``$bab$'', we divide the initial interval $[0, 1)$ into three sub-intervals $[0, 0.2)$, $[0.2, 0.7)$, and $[0.7, 1)$ and select $[0.2, 0.7)$ to represent ``$b$''. To encode the next symbol ``$a$'', we further divide $[0.2, 0.7)$ into $[0.2, 0.3)$, $[0.3, 0.55)$, and $[0.55, 0.7)$ based on the symbol probabilities and update the current interval to $[0.2, 0.3)$ which now represents ``$ba$''. This process continues until we reach the end of the message and obtain the final interval $[0.22, 0.27)$. We then select a fraction $q$ within the final interval that has the shortest binary representation (e.g., $q = (\text{.001111})_2$) as the message's code.

\subsection{Structure Learning}
Structure learning refers to the process of identifying correlations between columns to achieve better compression~\cite{DBLP:conf/kdd/GaoP16, DBLP:conf/sigmod/IlkhechiCGMFSC20, DBLP:conf/sigmod/BabuGR01, DBLP:conf/kdd/DaviesM99}. For example, the \texttt{gender} column is often highly correlated with the \texttt{name} column in a relation. As depicted in \cref{fig:correlation}, $80\%$ of \texttt{Taylor}s are female, while $50\%$ of \texttt{Alex}es are male. Instead of using static probabilities (e.g., $50\%$ male and $50\%$ female) for the \texttt{gender} column, we model its distribution using probabilities conditioned on the \texttt{name} column: $P_{\text{gender}}(\text{Female} | \text{Name} = \text{Taylor}) = 0.8$, $P_{\text{gender}}(\text{Female} | \text{Name} = \text{Alex}) = 0.5$. Then, the more common tuple (Female, Taylor) is mapped to a larger interval $[0, 0.56)$ with a short binary code $(.0)_2$, thus achieving better compression.

We use a Bayesian network (BN) to learn the best column ordering $S$ (e.g., \{\texttt{name}, \texttt{gender}\}) for compression. The output also includes a model set $M$, where each model is a probability distribution $P_x$ for column $x$ conditioned on the values of all the columns preceding $x$ in $S$. Determining the optimal ordering \( S \) is an NP-hard problem~\cite{DBLP:journals/ker/Parsons11a}. We, therefore, use a greedy algorithm~\cite{DBLP:conf/kdd/GaoP16} that selects the column that produces the smallest compressed size conditioned on the existing columns in $S$ for each iteration.

\section{Overview}
\revision{The goal of \work is to reduce the memory footprint of an in-memory OLTP database while imposing as small performance overhead as possible. To achieve this, \work must be able (1) to handle newly inserted tuples with unseen values efficiently and (2) to deliver low latency and high compression factor for individual tuples. For (1), we build semantic models that describe the values' (conditional) probability distributions instead of using static value dictionaries (\cref{sec:semantic-models}). For (2), we propose delayed coding that offers fast and fine-grained encoding/decoding with near-entropy compression (\cref{sec:delayed-coding}). \work is optimized for single-tuple retrieval. A larger compression granularity may improve the overall compression factor, but it introduces decompression overhead for point accesses common in OLTP workloads.}

\work sits above the table storage to compress and decompress tuples while remaining transparent to the execution engine of an OLTP database. When the execution engine inserts a tuple into a relation, \work compresses that tuple before sending it to the table storage. Upon receiving a tuple-fetching request, \work retrieves the compressed tuple from storage and decompresses it. The execution engine then consumes the tuple and executes the query without being aware of \work.

As shown in~\cref{fig:overview}, \work consists of three components: Semantic Learner (SL), Attribute Encoder (AE), and Tuple Encoder (TE). 
Specifically, the SL determines the compression ordering for the columns using structure learning and generates conditional probability models for the AE.
\revision{When the tables are small, \work leaves them uncompressed. SL is triggered when the size of a table reaches a predefined threshold (default: $2^{16}$ rows).}
For compression, the AE takes in a tuple and translates the value of each attribute into an interval in $[0, 1)$ according to the models from SL. The AE then sends the sequence of intervals to the Tuple Encoder which uses delayed coding to produce a compressed record with a near-optimal size. For decompression, the tuple is first decoded into 16-bit codes at TE. The AE then invokes the \texttt{Inv-Translate} (which refers to the probability models) to recover each tuple value.

\textbf{Semantic Learner} approaches the optimal column compression ordering $S$ and a set of models $M$ through a greedy algorithm of structure learning, as described in \cref{sec:structure-learning}. To speed up the structure learning on large tables, we perform the algorithm on a set of randomly selected tuples from the table. Once the column ordering $S$ is obtained, the SL further scans the full table to generate accurate conditional probability models $P_x \in M$. $P_x$ is implemented as an unordered map from each value combination of the proceeding attribute models to a probability distribution $p_i$ of attribute $x$. We refer to $p_i$ as a semantic model. Semantic models can compress values unseen before because they estimate the value distribution rather than statically mapping values to codes in a dictionary. We introduce two fundamental types of semantic models optimized for decompression speed in ~\cref{sec:semantic-models}.

\textbf{Attribute Encoder} converts each value into an interval and vice versa according to the semantic models. \texttt{Translate} maps an attribute value $v$ to an interval $[l, r)$, $0 \leq l < r \leq 1$ (symbol-to-interval), while \texttt{Inv-Translate} takes in a code $s \in [l, r)$ and recovers the attribute value $v$ (code-to-symbol). \texttt{Inv-Translate} is critical to the decompression performance. Classic arithmetic coding performs a binary search to determine the matching interval $[l, r)$ for a code $s$ with a time complexity of $O(\log N)$ where $N$ is the number of unique values in a column \cite{MacKay2003InformationTheory}. We optimize this procedure to constant time in \work, as detailed in~\cref{sec:discrete}.

\textbf{Tuple Encoder} receives a sequence of intervals representing each value within a tuple and compresses them into a block of 16-bit integers using delayed coding. At a high level, delayed coding uses a 16-bit unsigned integer $s_{\text{int}}$ to encode each interval $[l, r)$ such that $s_{\text{int}}/2^{16} \in [l, r)$. These integers are selected judiciously so that some integer codes are stored implicitly using the redundant information of the other integer codes. In this way, delayed coding not only supports fast decoding (because of the fixed-length codes) but also achieves near-entropy compression. We introduce delayed coding in detail in~\cref{sec:delayed-coding}.

\section{Semantic Models}
\label{sec:semantic-models}
A semantic model maps a value to an interval based on the estimated (conditional) probability distribution of the values within a column. In this section, we first introduce two fundamental models for discrete/categorical columns and continuous/numeric columns, respectively. We then show in~\cref{sec:composite} how to construct models for other data types (e.g., string) using the fundamental models.

%In this section, we introduce two foundational semantic/attribute models: the discrete (categorical) model and the continuous (numeric) model. We further build complex models based on them. 

\subsection{Discrete/Categorical Model}
\label{sec:discrete}

As in classic entropy encodings, we construct the semantic model for a discrete/categorical column by counting the frequency of each symbol (i.e., value) and computing their cumulative distribution function (CDF). Each symbol is then mapped to its corresponding probability interval on the CDF. For example, the semantic model for column $\{a, b, b, a, c, b, b, b\}$ is $\{a \leftrightarrow [0, 0.25), b \leftrightarrow [0.25, 0.875), c \leftrightarrow [0.875, 1)\}$. \texttt{Inv-Translate} a code (e.g., $(.01)_2$) back to its symbol (e.g., ``$b$'') requires a binary search to find the interval that contains the code. The logarithmic complexity slows down the decompression, especially when the number of distinct values of a categorical column is large. We, therefore, propose a constant-time algorithm for the \texttt{Inv-Translate} function, inspired by the alias method~\cite{kronmal1979alias}.

\begin{figure}[t!]
    \centering
    \includegraphics[width=\linewidth]{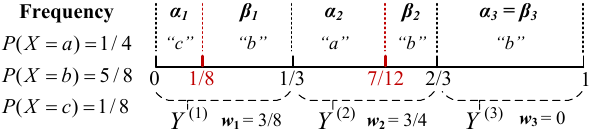}
      \vspace{-4ex}
    \caption{Interval Allocation By Pairing Symbols \textnormal{- There are three interval pairs  $\{Y^{(1)}, Y^{(2)}, Y^{(3)}\}$. In each pair, two symbols $\{(\alpha_N, \beta_N)\}$ and the symbol boundary $\{w_N\}$ are saved.}}
    \label{fig:alias_method}
      \vspace{-2ex}
\end{figure}

\textbf{Constant-Time Inv-Translate.}
Let $\pi_1, \pi_2, \cdots, \pi_N$ be the interval length (i.e., probability) of each of the $N$ symbols. Given a code $0 \le s < 1$, if $\pi_1 = \pi_2 = \cdots = \pi_N$, then $s$ belongs to the $\lfloor s \cdot N \rfloor$th interval. Therefore, the key to achieving constant-time \texttt{Inv-Translate} is to ``create'' a uniform distribution. The core idea of the algorithm is to pair the intervals so that the combined probability of each pair forms a uniform distribution.

%\textbf{Constant-Time Complexity of Inv-Translate.} Inspired by the alias method \cite{kronmal1979alias}, we have crafted an algorithm that achieves constant-time code-to-symbol translation. The core idea of our algorithm is to pair symbols, such that each pair has an equal total probability (interval length). This equalization allows us to map codes to symbols as if they were uniformly distributed. 

Suppose we want to create $N$ pairs, each having a combined interval length of $1/N$. At each iteration of the algorithm, we pair the shortest interval with the longest one in the remaining intervals. When their combined interval length is greater than $1/N$, we split the longer interval in two and put the exceeded part back into the interval collection for further pairing. The algorithm terminates when there is $\le 1$ interval left in the collection. \cref{fig:alias_method} shows an example where the interval for symbol ``$b$'' is divided and mapped to three pairs. The following theorem proves the validity of this algorithm for any discrete probability distribution.

% Suppose we intend to create $N$ pairs, each having a length of $1/N$. We pair symbols in such a way that their combined probability equals $1/N$.  A challenge arises when the sum of the probabilities of any two symbols exceeds $1/N$. We address this by dividing one of the symbols' interval into two smaller segments, ensuring the sum of their lengths matches the original probability. An example of this is the treatment of the symbol ``1'' demonstrated in \cref{fig:alias_method}, which is split into three segments to match its probability. The algorithm progresses by pairing symbols with the smallest and largest probabilities, subsequently splitting the interval of the larger probability. This process repeats until all symbols/intervals are paired and appropriately adjusted. \cref{thm:alias-method} validates that this approach applies to any discrete probability distribution.

\begin{thm}
Every probability vector  $\pi_1, \cdots, \pi_N$, can be expressed as an equiprobable mixture of $N$ two-point distributions. That is, there are $N$ pairs of integers $(\alpha_1, \beta_1)$, $\cdots$, $(\alpha_N, \beta_N)$ and probabilities $w_1, \cdots, w_N$ such that
\begin{equation*}
    \pi_i = 1/N \cdot \sum_{j = 1}^N (w_j \mathds{1}_{\{\alpha_j = i\}} + (1 - w_j) \mathds{1}_{\{\beta_j = i\}}) = 1/N \cdot \sum_{j = 1}^N Y^{(j)}_i.
\end{equation*}
for $1 \leq i \leq N$, where $Y^{(1)}$,~$\cdots$,~$Y^{(N)}$ are two-point distributions.
\label{thm:alias-method}
\end{thm}
\begin{proof}
See Appendix B.
\end{proof}

The constant-time \texttt{Inv-Translate} function is presented in~\cref{alg: inv_translate}. Given the binary mixtures $\{Y^{(j)}\}$ with parameters $\{(\alpha_{N}, \beta_{N})\}$ and $\{w_{N}\}$ defined in~\cref{thm:alias-method}, \texttt{Inv-Translate} first computes the index (i.e., $j = \lfloor s \cdot N \rfloor + 1$) of the binary mixture $Y^{(j)}$ that contains the input code $s$. Then the function determines the code's position within $Y^{(j)}$ and returns the corresponding symbol.

% The constant-time code-to-symbol function is given in Algorithm 1. It begins by identifying the index of the binary mixture $Y^{(j)}$, and then determines how the input number fits within $Y^{(j)}$'s probability boundaries to accurately return the corresponding symbol. 

\begin{algorithm}[t!]
\caption{Inv-Translate}
\label{alg: inv_translate}
\DontPrintSemicolon{}
\SetKwProg{Fn}{Function}{:}{end}
\SetKwFunction{InvTranslate}{Inv-Translate}
\text{Given $\{(\alpha_{N}, \beta_{N})\}$ and $\{w_{N}\}$ defined in~\cref{thm:alias-method}}.
\Fn{\InvTranslate{s}}{
     $c = s \cdot N$            $\ \ \ \ \ \ \ \ \ \ \ $\tcc{$s \in [0, 1)$}
     $j = \lfloor c \rfloor + 1$    $\ \ \ \ \ \ \ $\tcc{Determine the index of $Y$}
     $q = c - \lfloor c \rfloor$  $\ \ \ \ \ \ \ $\tcc{Get the position in $Y^{(j)}$}

     \KwRet ($q < w_j$)$\ $\text{?}$\ \ $$\alpha_j: \beta_j$\;
}
\end{algorithm}

\subsection{Continuous/Numeric Model}
\label{sec:continuous}
Previous semantic compression algorithms use a bisection method \cite{witten1987arithmetic, DBLP:conf/kdd/GaoP16} to compress continuous values such as floating-point numbers. This approach, however, generates many low-entropy intervals, thus affecting the efficiency of the subsequent Tuple Encoding. \work, therefore, introduces a novel two-level quantization model for continuous-value columns. This model not only supports arbitrary precision to guarantee lossless compression but also leverages the distribution skew in the column for better compression.

% Compressing floating-point numbers is traditionally challenging. Previous semantic compression uses a bisection method \cite{witten1987arithmetic, DBLP:conf/kdd/GaoP16}. This method generates many low-entropy intervals, putting strain on the subsequent encoding. \work introduces a novel model to address this challenge. We observe that numeric values of the same column, often fall within similar ranges and share common prefix digits. Leveraging this, we propose quantizing these values into disjoint ranges based on their prefixes. 

\textbf{Two-Level Quantization Model. } The first-level quantization is based on an equi-width histogram of the values in a column. The goal at this level is to maximize compression by assigning larger intervals (i.e., shorter codes) to more frequent value ranges. Specifically, we divide the values into a predefined $T$ (e.g., $T = 512$) disjoint value ranges, each having a bucket width of $w = (\hat{v}_{\text{max}} - \hat{v}_{\text{min}})/T$, where $\hat{v}_{\text{max}}$/$\hat{v}_{\text{min}}$ is the estimated (or obtained directly from table statistics) maximum/minimum value of the column. We then obtain the frequency of each bucket by scanning the column once, and we assign each bucket $i$ an interval $[l_i, r_i)$ proportional to its frequency, similar to the categorical model.

To guarantee a lossless compression (i.e., to distinguish between the values within a bucket), we apply a second-level quantization where we divide the value range of the bucket equally into $G$ segments so that the width of each segment is smaller than or equal to the column's required precision $p$. We set $p=10^{-7}$ and $p = 10^{-17}$ for the \texttt{float} and \texttt{double} types, respectively. Besides the bucket's interval $[l_i, r_i)$, each segment $j$ in the bucket is assigned another interval $[l_{{\epsilon}_j}, r_{{\epsilon}_j})$ with an equal length of $1/G$.
\revision{If the user specifies a precision requirement for a \texttt{float}/\texttt{double} column (e.g., $2$ decimal places), we enable lossy compression by adjusting $p$ accordingly (e.g., $p = 10^{-2}$) to achieve better compression.}

Given a value $\hat{v}_{\text{min}} \leq v < \hat{v}_{\text{max}}$, the \texttt{Translate} function computes its first-level bucket index $i$ and its second-level offset $j$ according to the value-range division because $v$ can be uniquely decomposed as \(i \cdot w + j \cdot p \leq v - \hat{v}_{\text{min}} < i \cdot w + (j + 1) \cdot p\), where \(j \cdot p < w\). $v$ is then converted into two intervals: $[l_i, r_i)$ and $[l_{{\epsilon}_j}, r_{{\epsilon}_j})$. In cases where $v$ is an outlier (i.e., $v < \hat{v}_{\text{min}}$ or $v \ge \hat{v}_{\text{max}}$), the algorithm falls back to the slow traditional bisection method. Such a fall back has negligible impact on performance because outliers are usually rare. To recover a value (a value range $\le$ the column precision $p$, to be precise), we invoke the \texttt{Inv-Translate} function twice on $[l_i, r_i)$ and $[l_{{\epsilon}_j}, r_{{\epsilon}_j})$, respectively.

% \textbf{Mapping: Numeric Value $\rightarrow$ Intervals.} The \texttt{Translate} function computes the bucket index \(I\) and the offset \(\varepsilon\) for a given value \(v_i\). Assuming \(v_i\) is not an outlier, i.e., \(\hat{v}_{\text{min}} \leq v_i < \hat{v}_{\text{max}}\),  it can be uniquely decomposed as \(I \cdot w + \varepsilon \cdot p \leq v_i - \hat{v}_{\text{min}} < I \cdot w + (\varepsilon + 1) \cdot p\), where \(\varepsilon \cdot p < w\), and both \(I\) and \(\varepsilon\) are integers. To represent the value $v_i$, \texttt{Translate} uses two distinct intervals: \([l_I, r_I)\) for the bucket index and \([l_{\epsilon}, r_{\epsilon})\) for the offset. In cases where \(v_i\) is an outlier, the algorithm resorts to the traditional bisection method. This slow method does not significantly reduce the overall decoding speed due to the rarity of outliers.

% \textbf{Lossless Numeric Compression. } Function \texttt{Inv-Translate} is called for two times to retrieve a numeric value from intervals \([l_I, r_I)\) and \([l_{\epsilon}, r_{\epsilon})\), resulting in a constant time decoding for non-outliers. The result value is rounded to a specified precision $p$, enabling lossless compression by adjusting this precision as needed. For numeric columns, we set $p$ higher than the column's required precision. For instance, for a C++ float, which can represent up to 7 decimal places, the default precision is set to $p=10^{-7}$. For double columns, capable of representing up to 17 decimal places, we set the precision to $p = 10^{-17}$.

\subsection{Composite Models}
\label{sec:composite}

%\subsection{Integrated Modeling}
%\label{sec:integrated}

\begin{figure}[t!]
  \centering
  \includegraphics[width=\linewidth]{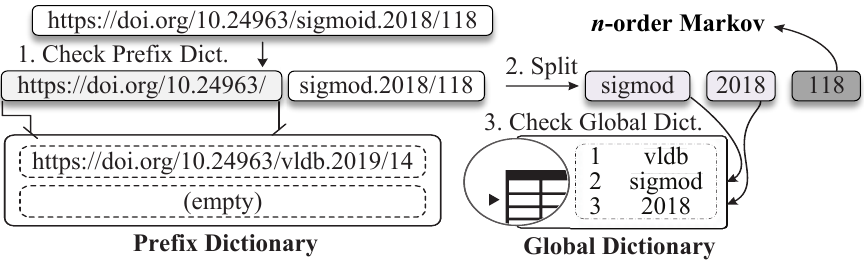}
  \vspace{-4ex}
  \caption{String Model - \textnormal{The URL sample is from the \textit{DBLP}.}}
  \label{fig:string}
    \vspace{-2ex}
\end{figure}

Our foundational models can be combined to compress complex attributes. We implement a string model as the example, as shown in \cref{fig:string}. It includes a prefix dictionary and a global dictionary. For words not covered by either dictionary, we use a Markov model to encode them letter-by-letter \cite{zpaq}. 

\textbf{Efficient \work Integration with Intervals.} The prefix dictionary compresses strings with similar prefixes and keeps a queue of the latest $K$ (e.g., $K=4$) strings. Each string is analyzed to find the index $i$ (ranging from 0 to $K-1$) of a previous string in the queue that shares the longest common prefix, and the count $h$ (an indefinite integer) of identical characters. We use a categorical model and a numeric model to estimate $i$ and $h$ distributions, respectively. Given a new string, the two models output intervals representing $i$ and $h$. This approach of interval-based representation integrates smoothly with the \work Framework.

\textbf{Adaptive Base Models for Enhanced Compression.} Following the prefix dictionary, the remaining substring is split into words using delimiters like spaces and commas. This process involves two models: (1) using a numerical model to count the words in the given substring, and (2) using a categorical model to identify each delimiter. Our model can leverage the skewed distribution of word count and delimiter usage patterns for compression. For example, most sentences have 3-10 words, and the space character is the most frequent delimiter. The final technique is the global dictionary, implemented as a categorical model. It stores words that frequently occur in sentences and is used for dictionary encoding.

Besides the string model, we also design two models: one for encoding JSON collections and another for time-series column encoding, with the latter utilizing the Autoregressive Moving Average (ARMA) \cite{box2015timeseries}. When applied to the data set \textit{Jena Climate} \cite{Jena-climate}, our time-series model achieved a 38\% better compression factor than our standard numeric model.

\section{Delayed Coding}
\label{sec:delayed-coding}
The attribute encoder generates a series of intervals, which can be encoded into a bit stream by the arithmetic coding. However, it is slow due to its variable-length codes and extensive floating-point calculations. We propose the delayed coding to address these issues. For ease of understanding, we assume that every symbol can be represented by a single, continuous interval. We relax the constraint in \cref{sec: noncontinuous_intervals}.

\subsection{Probability Representation}
\revision{We begin by investigating fast algorithms to represent and compute with probability intervals.}
Arithmetic coding is slow due to the many interval product operations required. Recall the example in \cref{sec: arithmetic coding}, when entering a sub-interval of the current interval based on the next symbol's probability, we need to compute a product of intervals. If we use a floating-point number to represent the probability, the interval product $\boxtimes$ is defined as: 
\begin{equation}
    [l_a, r_a) \boxtimes [l_b, r_b) \coloneqq [l_a + (r_a - l_a) \cdot l_b, l_a + (r_a - l_a) \cdot r_b),
    \label{equ:interval_product}
\end{equation}
where $0 \leq l_a < r_a \leq 1$, $0 \leq l_b < r_b \leq 1$. However, this design leads to floating-point calculations and the risk of floating-point underflow. An alternative is to use two integers $U, d$ to represent the probability $U / 2^d$. In this way, checking for underflow is simply inspecting the exponent of the denominator  $d$.

\textbf{Integer-based Probability Intervals.} In \work, a probability is presented as a 16-bit integer $U$, which logically represents the probability $U / 2^{16}$. In this paper, we use $[L, R)$, where $0 \leq L, R \leq 2^{16}$ are integers, to logically represent the interval $[L / 2^{16}, R / 2^{16})$. We choose to use 16 bits because using shorter integers increases the decoding overhead while using 32-bit or 64-bit integers must handle integer overflow during multiplication. For an interval whose length is smaller than $1/2^{16}$ (e.g. $[0, 1/2^{32})$), we use the product of two or more intervals to represent it, i.e. 
\begin{equation*}
    [0, 1/2^{32}) = [0, 1/2^{16}) \boxtimes [0, 32768/2^{16}) \rightarrow [0, 1), [0, 32768).
\end{equation*}

\begin{algorithm}[t!]
\caption{Inv-Translate with the Integer Probability}
\label{alg: inv_translate_integer}
\DontPrintSemicolon{}
\SetKwProg{Fn}{Function}{:}{end}
\SetKwFunction{InvTranslate}{Inv-Translate}
\text{Given $m$, $\{(\alpha_{N}, \beta_{N})\}$, $\{w_{N}\}$ defined in \cref{thm:alias-method}, where }
\text{$N = 2^{m}, m \in \mathbb{N}_+$, and $w_i$ is an 16-bit integer, for $i = 1, \cdots, N$. \;}
\Fn{\InvTranslate{s}}{
    $j \leftarrow s \ShiftRight (16 - m)$ $\ \ \ \ $\tcc{The higher $m$ bits }
    $Q \leftarrow s \ \&\  (2^{16-m} - 1)$ $\ \ $\tcc{The lower 16-$m$ bits}
    \KwRet $(Q < (w_j \ShiftRight m))$\ $\text{?}$\ \ $\alpha_j$\ : $\beta_j$\;
}
\end{algorithm}

\textbf{Inv-Translate Becomes Faster.} Although the code-to-symbol in \cref{alg: inv_translate} has constant time complexity, the ``floor'' operation slows it down. The integer-based probability can make the \texttt{Inv-Translate} faster, replacing the ``floor'' operation with bitwise operations. \cref{alg: inv_translate_integer} gives the new \texttt{Inv-Translate} algorithm with integer-based probability. Note that both the $\{w_N\}$ and the input $s$ in \cref{alg: inv_translate_integer}, using the integer-based probability, are 16-bit unsigned integers. Also, $N$ is always a power of two, a condition that can be satisfied by adding placeholder symbols with a frequency of 0 if necessary.

\subsection{Options of Fixed-length Codes}
\revision{We then introduce an algorithm to extract the code from an interval based on our integer-based probability representation and analyze the bits wasted in terms of code options.}
In \cref{sec: arithmetic coding}, we explained that for the final interval $[0.22, 0.27)$ in arithmetic coding, using just enough fractional digits to keep the number within this interval is sufficient for encoding. However, this results in variable-length codes, which can slow down decoding due to the inconsistent number of digits across intervals, requiring more branch predictions and checks \cite{said2004comparative}.

An approach to simplify the decoding process is to utilize a fixed number of bits, such as 16 bits, for encoding each interval. In this scenario, adopting integer-based probability, we define the bidirectional mapping of a semantic model as follows:
\begin{equation*}
\begin{aligned}
\text{symbol-to-interval}:&\ \{v_1, \cdots, v_N\} \rightarrow \{[L_i, R_i), 1 \leq i \leq N\},\\
\text{code-to-symbol}:&\ \{s \in \mathbb{N}, 0\leq s < 2^{16}\} \rightarrow \{v_1, \cdots, v_N\},\\
\end{aligned}
\end{equation*}
where $v_i$ denotes a unique symbol in this model, allocated with a disjoint interval $[L_i, R_i)$, $L_i$ being the lower bound and $R_i$ being the upper bound of the interval, for $1 \leq i \leq N$. By converting the interval $[0.22, 0.27)$ to $[14418, 17694)$ using integer-based probability, we can encode it using any 16-bit integer within this range. However, this method requires more than the 6 bits we use in the example of \cref{sec: arithmetic coding}. In fact, for an interval $[L, R)$, the probability of the event it represents is $(R - L) / 2^{16}$. Thus, its entropy is $16 - \log_2 (R - L)$ bits, and we waste $\log_2 (R - L)$ bits if using 16 bits to encode it. Notice that we have $R - L$ code options for this interval. The selection of a specific code option itself carries information. Since any of these options can represent the interval $[L, R)$, we can use the options to represent another interval partially. 

\subsection{Problem Formulation}
\revision{Leveraging the concept of code options, we formalize the code extraction problem as follows.}
Given a series of intervals $[L_1, R_1), \cdots,$ $[L_n, R_n)$ with $n > 0$ and $0 \leq L_i < R_i \leq 2^{16}$, we can select a 16-bit integer $s_i \in [L_i, R_i)$ to encode each interval. 
%Now, given a distinct interval $[L^*, R^*)$, can we represent it using the sequence $(s_1, \cdots, s_n)$? 
Then, can we represent a distinct interval $[L^*, R^*)$ using a sequence $(s_1, \cdots, s_n)$? 

Consider a special case where each interval is of length 2 (i.e., $R_i - L_i = 2$), we have two encoding options per interval: either $L_i$ or $L_i + 1$. These options can uniquely represent their respective intervals. Using this binary decision for each interval, we can represent parts of the distinct interval $[L^*, R^*)$. To fully encode $[L^*, R^*)$, we need a sufficient number of intervals. If we have at least 16 intervals, matching the 16 bits of $L^*$, complete representation is possible. The encoding rule is simple:  Let the code $s_i = L_i$ if the $i$-th bit of $L^*$ is 0; otherwise, $s_i$ takes the value $L_i + 1$. This way, the sequence $(s_1, \cdots, s_n)$ encodes both the series of intervals and the interval $[L^*, R^*)$, leveraging the binary encoding options of each interval.

Consider a general case where each interval's length \( k_i = R_i - L_i \) varies within the range \([1, 2^{16})\). This means each interval \([L_i, R_i)\) has \( k_i \) coding options: \(\{L_i, L_i + 1, \cdots, L_i + (k_i - 1)\}\). Any of these codes can uniquely represent its interval. Moreover, these codes can represent the distinct interval \([L^*, R^*)\) partially as well. In this case, the \(i\)-th interval offers a digit in a base \( k_i \) system. With \( n \) such intervals, they collectively form a mixed radix (base) numeral system \cite{fraenkel1985systems}, with bases \(\{k_n\}\). Assuming \(\prod_{i=1}^{n} k_i \geq 2^{16}\), the 16-bit number \(L^*\) can be converted into this mixed-base system as follows:
\begin{equation}
    a_i = L^*\ \%\ k_i, \ \ \ \ \ \ L^* = L^*\ /\ k_i ,
    \label{equ: base-conversion}
\end{equation}
for \(i = n, \cdots, 1\) in a loop. The resulting value is \(a_1 a_2 \cdots a_n\). Setting \(s_i\) to \(L_i + a_i\) gives a valid code, as \(a_i < k_i\) ensures \(L_i + a_i < R_i\). Hence, the sequence \((s_1, \cdots, s_n)\) effectively encodes both the series of intervals and the interval \([L^*, R^*)\).

\textbf{Example: 3-Digit Mixed Radix Numeral System. } Given intervals $[1, 4)$, $[2, 6)$, $[3, 10)$, forming a 3-digit mixed radix numeral system with bases $(3, 4, 7)$. This system can encode up to $3 \times 4 \times 7 = 84$ distinct states. Assume we use 4 bits to encode each interval. To encode the interval $[13, 14)$, we select $x = 13$ from this range. Decomposing 13 with the bases $(3, 4, 7)$  using \cref{equ: base-conversion} gives $13 = \underline{0} \times (4 \times 7) + \underline{1} \times 7 + \underline{6}$, leading to the indices $(0, 1, 6)$. To determine $s_1$, we select the value at index 0 in the first interval $[1, 4)$, resulting in $s_1 = 1$.  Similarly, $s_2 = 3$ and $s_3 = 9$. Therefore, the encoded bit stream for the four intervals is $($0001 0011 1001$)_2$.

\begin{figure*}[t!]
  \centering
  \includegraphics[width=\linewidth]{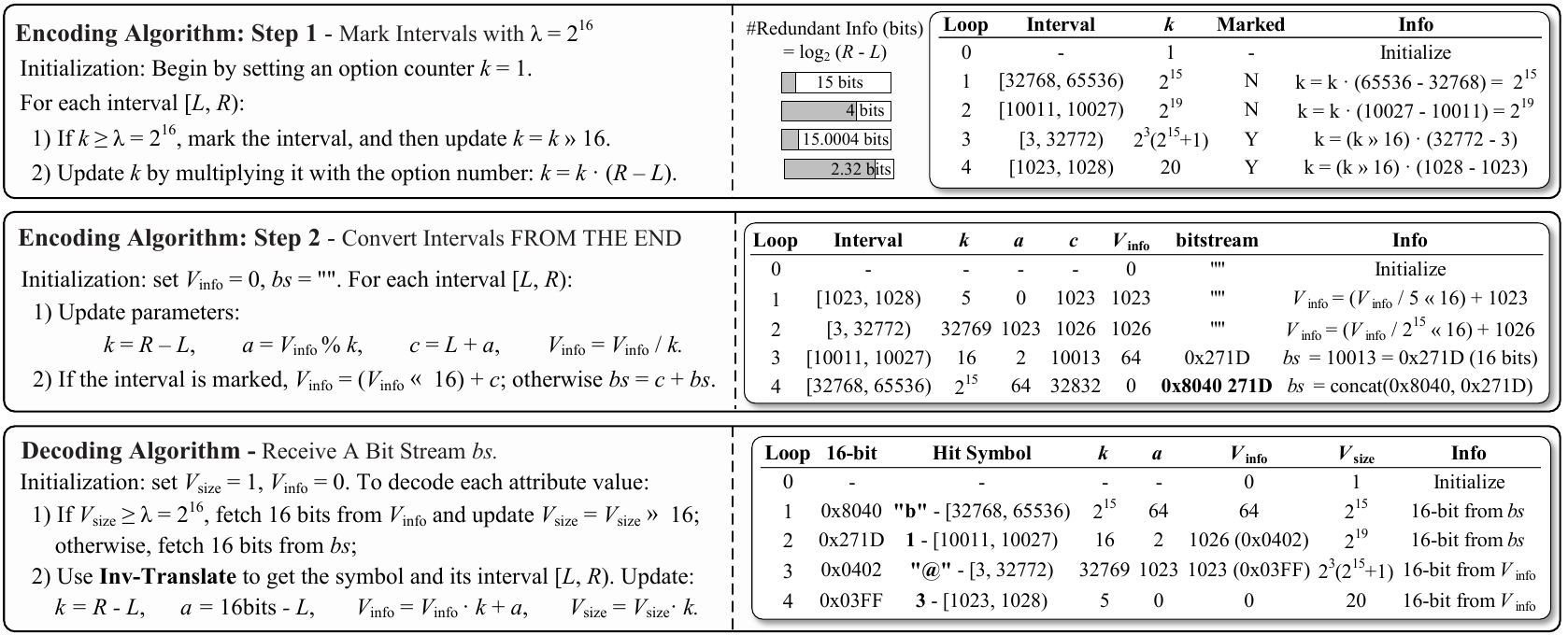}
  \vspace{-4ex}
  \caption{Delayed Coding \textnormal{- Given a tuple (``b'', 1, ``@'', 3), the attribute encoder translates it into intervals: [32768, 65536), [10011, 10027), [3, 32772), [1023, 1028). These intervals are then encoded and decoded as shown above.}}
  \label{fig:encode}
  \label{fig:decode}
  \vspace{-2ex}
\end{figure*}

\subsection{Encoding Procedure}
\begin{figure}[t!]
    \centering
    \includegraphics[width=0.9\linewidth]{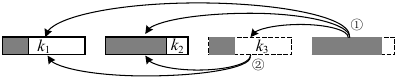}
      \vspace{-2ex}
    \caption{Recursive Encoding - \textnormal{First, the last interval is encoded by the numeral system with base $(k_1, k_2, k_3)$. Then, the last-second interval is encoded by the numeral system with base $(k_1, k_2)$. }}
    \label{fig:numeral_sys}
      \vspace{-2ex}
\end{figure}
\revision{Inspired by the mixed radix numeral system introduced above, the encoding procedure of delayed coding is essentially transforming decimal numbers into mixed radix numbers.} 
The encoding of delayed coding has two steps. First, we mark all intervals that can be represented by their former intervals' options. An interval can be marked if and only if the current option number is larger than $\lambda$ (it takes $2^{16}$ by default). Second, for each marked interval, we convert its 16-bit code into a mixed-radix number,
represented using code options of the intervals that precede it. 

\textbf{Step 1: Mark Intervals. }In \cref{fig:encode}, we illustrate the encoding process of a tuple (``b'', 1,``@'', 3) using four categorical models that convert the tuple into intervals. We use an option counter $k$, initially set to one, to track redundant information. The 1st interval provides $2^{15}$ options, updating $k$ to $2^{15}$. Next, the 2nd interval cannot be marked because the current option counter $k \leq \lambda = 2^{16}$. Note that we use fixed-length code (i.e., 16-bit) to encode each interval, and the current numeral system cannot represent a 16-bit integer. The 2nd interval increases $k$ to $2^{19}$ due to its option of 16. Currently, it is enough to mark the third interval, it consumes $2^{16}$ options and updates $k = k / 2^{16} = 2^{3}$. The 3rd interval contributes $32769$ options, updating $k = 2^{3} (2^{15}+1)$. The marking process continues; finally, the last two intervals are marked and will be transformed into mixed radix numbers, represented by their preceding intervals.

\textbf{Step 2: Convert Intervals From the End.} We convert each marked interval into a mixed-radix number in a recursive manner, as illustrated in \cref{fig:numeral_sys}. Specifically, the last two intervals are marked in the marking step. First, we convert the rightmost interval into a mixed-radix number using bases \((k_1, k_2, k_3)\). Then, the last-second interval holds the partial code of the last interval. Since the last-second interval is also marked, it is converted into a mixed-radix number with bases \((k_1, k_2)\). This approach highlights the necessity of processing intervals from the end.

We use a 64-bit integer $V_{\text{info}}$ to store the codes for marked intervals temporarily, starting with $V_{\text{info}} = 0$. We use a loop to process each interval $[L, R)$ from the end. For each step, the decimal number to be converted is $V_{\text{info}}$, and the current interval  $[L, R)$ provides $k = R- L$ options, as a base-$k$ digit. We compute the digit value $a$ and the left decimal number, using \cref{equ: base-conversion}:
\begin{equation}
\begin{aligned}
     a = V_{\text{info}}\ \%\ k, \ \ \ \ \ \ V_{\text{info}} =  V_{\text{info}}\ /\ k. \\
\end{aligned}
\label{equ:decomposition}
\end{equation}
Therefore, the 16-bit code for the interval $[L, R)$ is computed as $c = L + a$, i.e., we use the code options of it to store a digit value $a$. For the first processed interval, we get $a = 0$ because $V_{\text{info}} = 0$, but $V_{\text{info}}$ can be updated. If the interval $[L, R)$ is marked, $V_{\text{info}} = V_{\text{info}} \cdot k + c$. Otherwise, output the 16-bit code $c$ to the bit stream. The loop continues; finally, we encode the four intervals into 4 bytes.

\subsection{Decoding Procedure}
\revision{Conversely, the decoding procedure transforms the mixed radix numerals back into decimal numbers.}
There are two sources of bits to decode a tuple: the bit stream or the virtual input $V_{\text{info}}$. The decoding of each symbol has three steps: (1) retrieve a 16-bit code from $V_{\text{info}}$ if the current option number  $V_{\text{size}}$ is larger than $2^{16}$; otherwise from the bit stream; (2) obtain the desired symbol and its interval $[L, R)$ by calling the \texttt{Inv-Translate} function; (3) Update $V_{\text{info}}$ and  $V_{\text{size}}$ accordingly.

The bottom part of \cref{fig:decode} shows the decoding process. We want to decode a tuple from the bit stream \texttt{0x8040 271D}, at first, $V_{\text{info}} = 0$ , and $V_{\text{size}} = 1$. For the first attribute value, We fetch 16 bits from the bit stream, getting \texttt{0x8040}. The function \texttt{Inv-Translate} receives it and returns the symbol ``b''. We use the symbol-to-interval mapping to determine its interval, resulting in  $[L, R) = [2^{15}, 2^{16})$. Next, we   compute the digit base $k = R - L = 2^{15}$, and the digit number $a = \text{16bits} - L = 64$. Using them, we update $V_{\text{info}}$, and  $V_{\text{size}}$: 
\begin{equation*}
\begin{aligned}
    V_{\text{info}} = k \cdot V_{\text{info}} + a, \ \ \ \ \ \ V_{\text{size}} = k \cdot V_{\text{info}}.
\end{aligned}
\end{equation*}
This is just the inverse formula of \cref{equ:decomposition}. It is necessary to use $V_{\text{size}}$ to record the amount of information in $V_{\text{info}}$. For instance, with $V_{\text{info}} = 1$, a $V_{\text{size}}$ of 4 results in two virtual bits $01_2$, while a $V_{\text{size}}$ of 8 yields three virtual bits $001_2$. The second interval decodes to the symbol ``1''. When $V_{\text{size}} = 2^{19}$ ($\geq \lambda = 2^{16}$), we fetch the next 16-bit code from the virtual input, resulting in \texttt{0x0402} for the third symbol, as shown in the decoding table Loop 3 of \cref{fig:encode}. This decoding process repeats for all symbols, getting (``b'', 1, ``@", ``3'').

\subsection{Modification for Non-Continuous Intervals}
\label{sec: noncontinuous_intervals}
\revision{Up to this point, we assume that each symbol is represented by a single continuous interval. In this section, we modify our algorithm by relaxing this constraint to allow a symbol to be represented by the union of multiple non-continuous intervals.}
A non-continuous interval example is shown in \cref{fig:alias_method}, where the symbol ``$b$'' is assigned to the interval $[1/8, 1/3) \cup [7/12, 1)$, or its integer representation $[8192, 21845) \cup [38229, 65536)$. 

Non-continuous intervals, offering the same number of options as continuous ones but with different option positions, require slight modifications in delayed coding. Take a non-continuous interval with two segments \([L^{(1)}, R^{(1)})\cup [L^{(2)}, R^{(2)})\). It provides \( k = (R^{(1)} - L^{(1)}) + (R^{(2)} - L^{(2)}) \) options. To store a number \( a \in [0, k]\) using these options, we modify the selection of the 16-bit code \( c \) as follows: 
\[
c = 
\begin{cases} 
L^{(1)} + a & \text{if } 0 \leq a \leq R^{(1)} - L^{(1)}, \\
L^{(2)} + a - (R^{(1)} - L^{(1)}) & \text{if } R^{(1)} - L^{(1)} \leq a \leq k.
\end{cases}
\] 
In other words, we choose the \textit{a}-th optional code of the symbol, regardless of its integer value. Decoding involves the reverse process. For a 16-bit code within this non-continuous interval, we retrieve the stored number \( a \) as follows: if $\text{16bits} \in [L^{(1)}, R^{(1)})$, then $a = \text{16bits} - L^{(1)}$. Otherwise, $a = \text{16bits} - L^{(2)} + (R^{(1)} - L^{(1)})$. 
In other words, the 16-bit code is the $a$-th item in this non-continuous interval. This method can be extended to manage intervals with more than two segments, generating two piecewise linear functions for the computation of $c$ and $a$. Importantly, this modification has no effect on the correctness and efficiency of delayed coding, as shown in the appendix.

\subsection{Fine Granularity Compression Effectiveness }
\label{sec: fine_granularity}
\revision{We show the effectiveness and optimality of delayed coding in this section.}
In \cref{fig:encode}, there are 20 unused options (i.e., $k = 20$) after the encoding, resulting in a waste of $\log_2 20 = 4.32$ bits. The number of wasted bits can be bounded by $\log_2 \lambda$ (note that we mark an interval once the number of options is larger than $\lambda$). \cref{thm:delayed_coding} shows that as the number of intervals grows, the effectiveness of delayed coding improves, approaching the entropy. 

\begin{thm}
\label{thm:delayed_coding}
Give a series of intervals $[L_1, R_1)$, $\cdots$, $[L_n, R_n)$, where $n \geq 1$, and $L_i, R_i$ are 16-bit integers less than or equal to $2^{16}$ for all $i$. Suppose delayed coding:
\begin{enumerate}
    \item Marks an interval if and only if the current option number is larger or equal to $\lambda$, where $\lambda \geq 2^{16}$.
    \item Encodes every $\zeta$ intervals as a bit stream, where $0 < \zeta \leq n$.
\end{enumerate}
Thus, the number of used bits $L^{n, \lambda, \zeta}$ is bounded by
\begin{equation*}
    L^{n, \lambda, \zeta} \leq n \cdot C + (\lfloor n / \zeta \rfloor + 1) \cdot \log_2 \lambda + n\cdot \log_2 (1 - 65535 / \lambda)^{-1},
    \label{loss}
\end{equation*}
where $n \cdot C$ is the entropy of all intervals. Further, for a sufficiently large $n$, by setting $\lambda = \zeta = n$, we have $L^{n, \lambda, \zeta}/ (n \cdot C) \to 1$ as $n\to +\infty$. 
\end{thm}
\begin{proof}
 See Appendix D.2.
\end{proof}
\textbf{Summary: } Delayed coding uses a fixed number of bits to encode each interval. It is based on the insight that altering the redundant information in an interval does not affect its symbol retrieval. \cref{thm:delayed_coding} reveals that delayed coding has a near-entropy compression factor with fine compression granularity.

\section{Compression Microbenchmarks}
\begin{table}[t!]
    \small
    \centering
    \caption{Data sets \textnormal{- Unmarked data sets come from Public BI Benchmark \cite{DBLP:conf/sigmod/VogelsgesangHFK18} or earlier semantic compression research \cite{DBLP:conf/kdd/GaoP16, DBLP:conf/sigmod/IlkhechiCGMFSC20}.}}
    \vspace{-2ex}
    \begin{tabular}{clrcc}
    \toprule
       \multicolumn{1}{l}{\textbf{Group}} & \textbf{Data sets}          & \textbf{\#Rows}     & \textbf{\#Cols}    & \textbf{Row Length}  \\
    \midrule
\multirow{3}{*}{\rotatebox[origin=c]{90}{Numeric}}     & Corel                                         & 68,040     & 93  & 820 byte\\
                             & Jena Climate \cite{Jena-climate} & 420,551       & 14            & 138 byte                     \\
                             & Cars                                          & 344,287    & 155 & 393 byte \\
    \midrule
\multirow{4}{*}{{\rotatebox[origin=c]{90}{Categorical}}} & Forest Cover                                  & 581,012    & 55  & 127 byte \\
                             & US Census 1990                                & 2,458,285  & 69  & 145 byte\\
                             & Food                                          & 5,216,593  & 5   & 22 byte\\
                             & Bimbo                                         & 20,259,279 & 12  & 54 byte\\
    \midrule
\multirow{3}{*}{\rotatebox[origin=c]{90}{String}}& Yale Languages                                & 5,762,082  & 30  & 284 byte\\
                             & Medicare                                      & 8,645,072  & 26  & 229 byte\\
                             & Arade                                         & 9,888,775  & 11  & 88 byte \\
    \bottomrule
    \end{tabular}
    \vspace{-2ex}
    \label{tab:datasets}
\end{table}

\begin{figure*}[t!]
  \centering
  \includegraphics[width=0.97\linewidth]{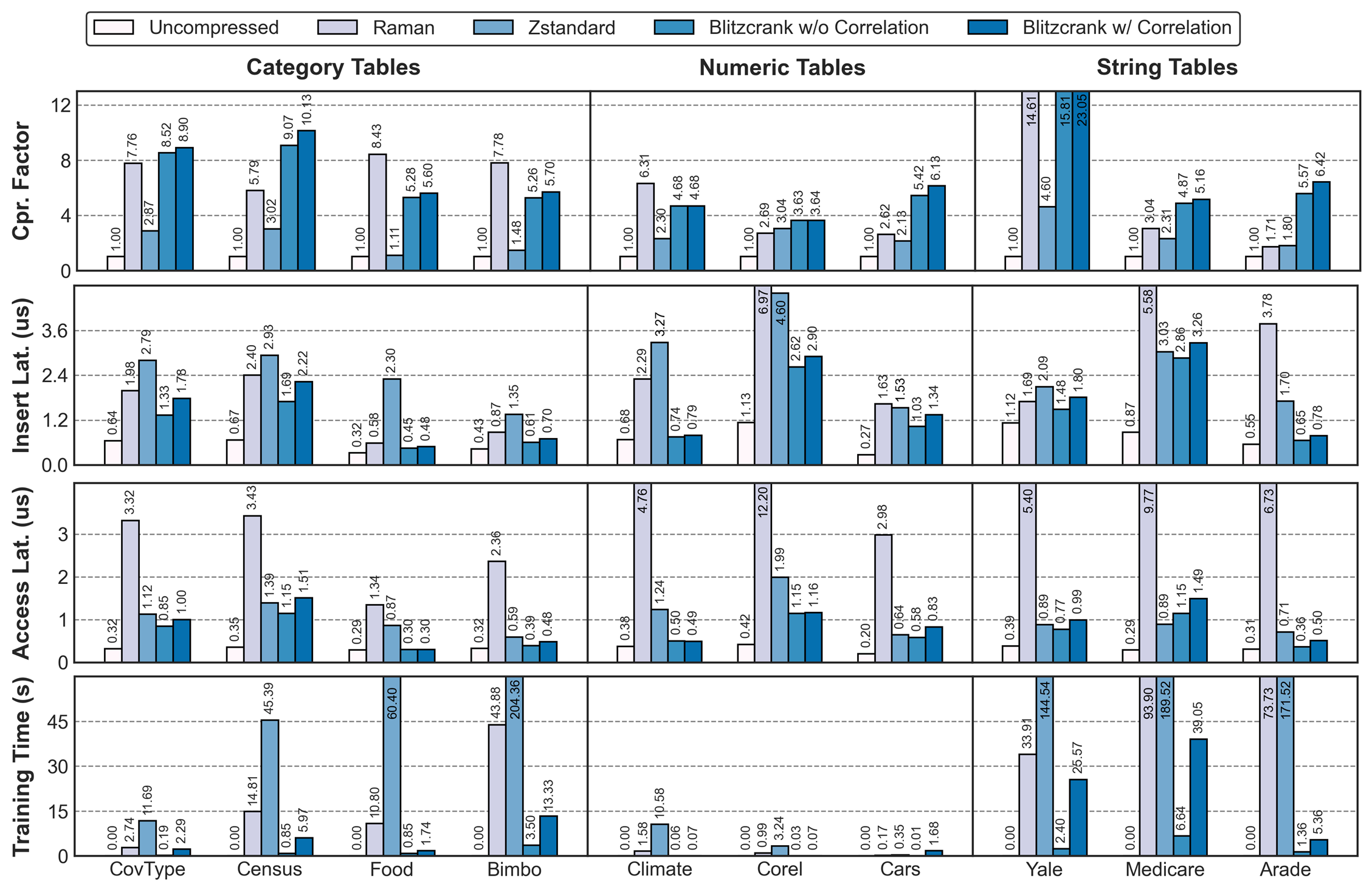}
  \vspace{-3ex}
  \caption{\revision{Compression Evaluation \textnormal{- We report the compression factor, insert/access latency, and training time of all compressors.}}}
  \vspace{-1.5ex}
  \label{fig:compression_ra}
\end{figure*}

We evaluate \work in the next two sections. First, using 10 real tables, we compare \work with modern compressors. This comparison focuses on compression factors and fast random tuple access from compressed storage (\cref{exp: compression_evaluation}). Then, we provide a breakdown of the \work structure learner (\cref{exp: sensitivity_to_sampling_number}). Following this, we compare delayed coding with asymmetric numeral systems (\cref{exp: entropy_coding}). Finally, we optimize the random access performance by analyzing the compression block size (\cref{exp: sensitivity_to_granularity}). 

\textbf{Baselines.} We evaluate \work against Zstandard \cite{DBLP:journals/rfc/rfc8878zstd} and Raman's approach \cite{DBLP:conf/vldb/RamanS06}: (1) Zstandard is a real-time compression system. It has a training mode, designed for compressing many small files. This mode creates a ``zstd-dictionary'' from all files and uses it to compress each file independently. \revision{We use the open-source Zstandard (v1.5.1) in C++, setting the ``zstd-dictionary'' capacity to the recommended 110 KB and using the default compression level.} (2) Raman's method \cite{DBLP:conf/vldb/RamanS06} focuses on tuple compression. It considers correlations between columns and combines Huffman coding and delta encoding to achieve a high compression factor.
\revision{We implemented Raman's approach in C++ using the default column ordering of an input table.}

We exclude DeepSqueeze \cite{DBLP:conf/sigmod/IlkhechiCGMFSC20} because it does not support high cardinality columns and is not open-sourced. We do not include FSST \cite{boncz2020fsst} and other lightweight techniques \cite{damme2017lightweight, abadi2006integrating, abadi2013design}, because they are not for row-stores. In the appendix, we also evaluate \work against \revision{the open-sourced (in C++)} Squish \cite{DBLP:conf/kdd/GaoP16} and Gzip \cite{DBLP:journals/tit/ZivL77gzip} for the table archive task. Our method is $20\times$ faster than Squish and offers $2\times$ higher compression factors compared to Gzip. 

\textbf{\work Setting. } \work samples $2^{15}$ tuples for structure learning, with detailed sensitivity analysis provided in \cref{exp: sensitivity_to_sampling_number}. For delayed coding, each tuple is individually encoded for the optimal access latency. We set $\lambda = 2^{16}$ to maximize the compression factor, as detailed in \cref{thm:delayed_coding}. We evaluate two \work variants: one utilizes column correlation for compression (\work w/ Correlation), while the other does not (\work w/o Correlation).

\textbf{Data Sets. } \cref{tab:datasets} shows the data sets. Besides the data sets from previous studies \cite{DBLP:conf/sigmod/VogelsgesangHFK18, DBLP:conf/kdd/GaoP16, DBLP:conf/sigmod/IlkhechiCGMFSC20}, we also use \textit{Jena Climate}, which consists of 14 time-series columns \cite{Jena-climate}. We classify each data set into categorical, numeric, or string types. Specifically, we calculate the proportion of each attribute type in the total data set size and select the type with the highest proportion as the representative group for that data set.

\textbf{Experimental Setup. } We use three metrics to measure compression performance: compression factor, throughput, and random access latency. The throughput represents the amount of data processed per unit of time for a given compression/decompression task. Random access latency is the time required to retrieve a random record. We conduct our experiments on a machine equipped with two $\text{Intel}^\text{\textregistered}$ Xeon 8375C (32 $\times$ 2 cores) and $512$ GB RAM. The disk we use is an $\text{Intel}^\text{\textregistered}$ SSD D5-P5530 (1 TB).
\revision{We use Debian GNU/Linux 11 and GCC 10.2 with \texttt{-O3} enabled.}
All microbenchmarks are conducted with a single thread.

\subsection{Compression Evaluation}
\label{exp: compression_evaluation}
We evaluate \work 
on in-memory tables (constructed using the data sets above)
with each tuple compressed separately.
The compressed tuples are organized using a primary-key index
(implemented using a simple C++ vector)
where the primary keys are monotonically increasing integers.
We use YCSB (workload C) with a Zipf distribution
to generate the random-access workloads \cite{cooper2010benchmarking}.
Specifically, for each data set, we first compress and insert 5 million tuples into the in-memory table
and then execute 1 million point queries, each involving decompressing a particular tuple.
We report the average latencies for compression-insertion and random access separately. 
We also record the size of each in-memory table after insertion to calculate the compression factor.
For each compressor, we first train its model over the corresponding data set if required by the algorithm.

\cref{fig:compression_ra} shows the results, including compression factor, latency, and training time, across various data sets on the x-axis. \work has the highest compression factor for 7/10 tables, and offers the lowest latency for 9/10 tables among all compressors. This is because \work models columns in a semantic way and uses fixed-length code for encoding. Raman's approach has the highest compression factor for the remaining 3/10 tables because tuples of these tables have low entropy; each tuple requires on average just a few bits for encoding (e.g., 2.6 bytes for a \textit{Bimbo}'s tuple). In this case, using fixed-length codes is less efficient. However, Raman's approach is slow for accessing tuples, because its variable-length code for each attribute leads to additional checks when decoding. Zstandard falls short for both the compression factor and the latency, because it relies on long contexts, at least 4KB, to build an effective dictionary. However, the length of a single tuple is insufficient to meet this 4 KB requirement. 

We advise using \work w/o Correlation in most cases.
%because it has a relatively high compression factor and low access latency.
Capturing column correlations improves the compression factor in the sacrifice of the access latency and the model training time. The training time is often exponential to the number of categorical columns, but a longer training time does not necessarily guarantee better performance. A detailed analysis of the semantic learner is given in \cref{exp: sensitivity_to_sampling_number}.

\begin{figure*}[t!]
    \centering
      \begin{minipage}[b]{0.455\textwidth}
        \centering
        \includegraphics[width=\linewidth]{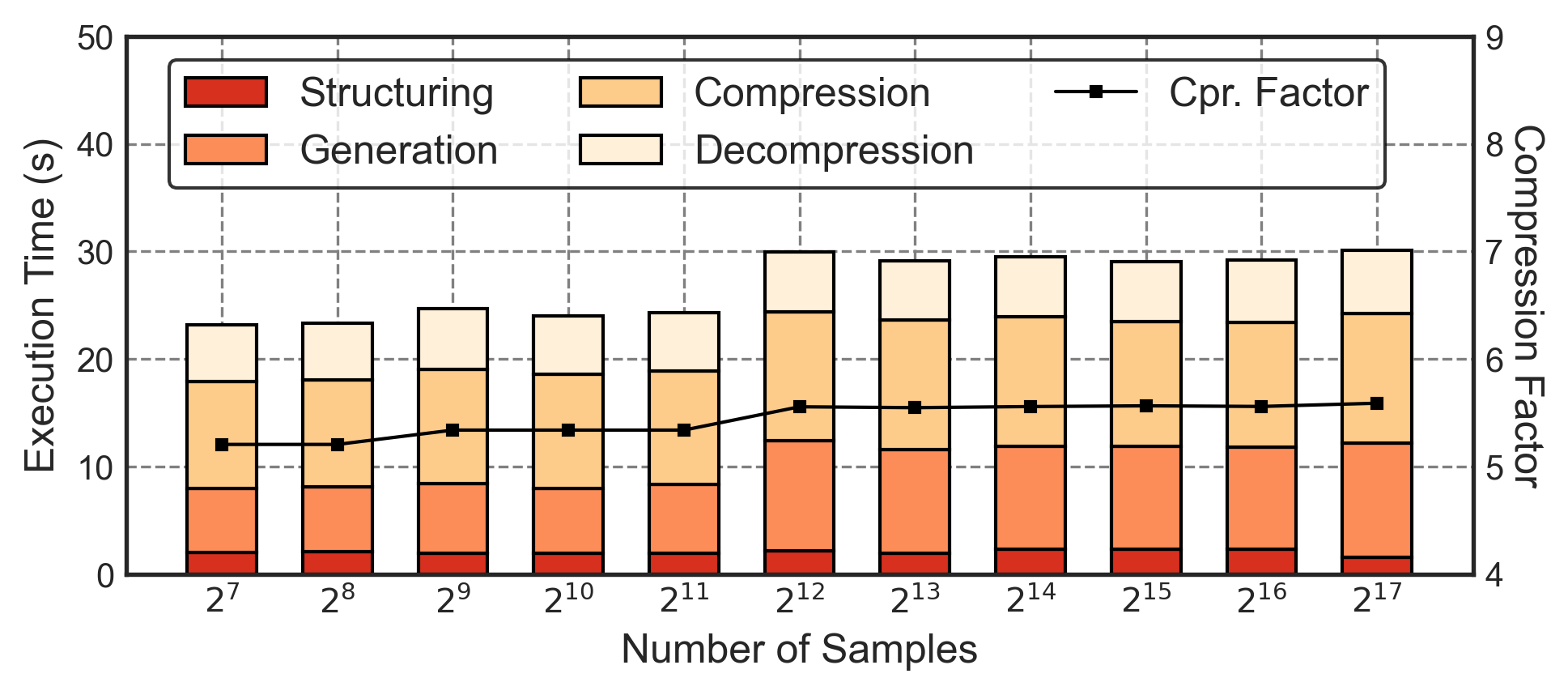}
        \vspace{-6.5ex}
        \caption*{(a) \textnormal{No Correlation Found Example: \textit{Bimbo}} }
      \end{minipage}
      \hspace{3ex}
      \begin{minipage}[b]{0.455\textwidth}
        \centering
        \includegraphics[width=\linewidth]{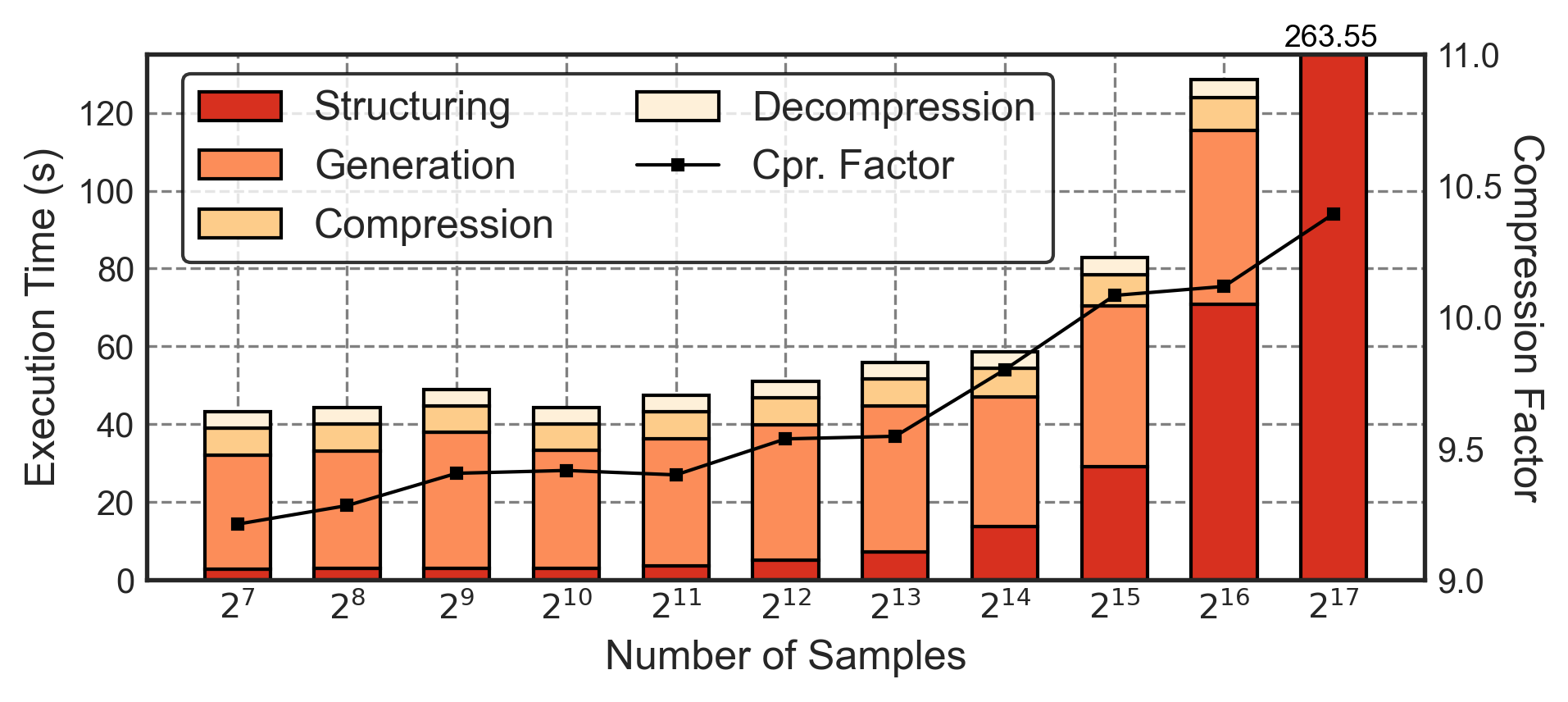}
        \vspace{-6.5ex}
        \caption*{(b) \textnormal{Correlation Found Example: \textit{Census}} }
      \end{minipage}
      \vspace{-2ex}
      \caption{\revision{Breakdown of \work Distribution Learner \textnormal{- Vary the \#samples and evaluate the performance.}}}
      \label{fig:sample number}
      \vspace{-2.5ex}
\end{figure*}

\begin{figure}[t!]
    \centering
      \includegraphics[width=\linewidth]{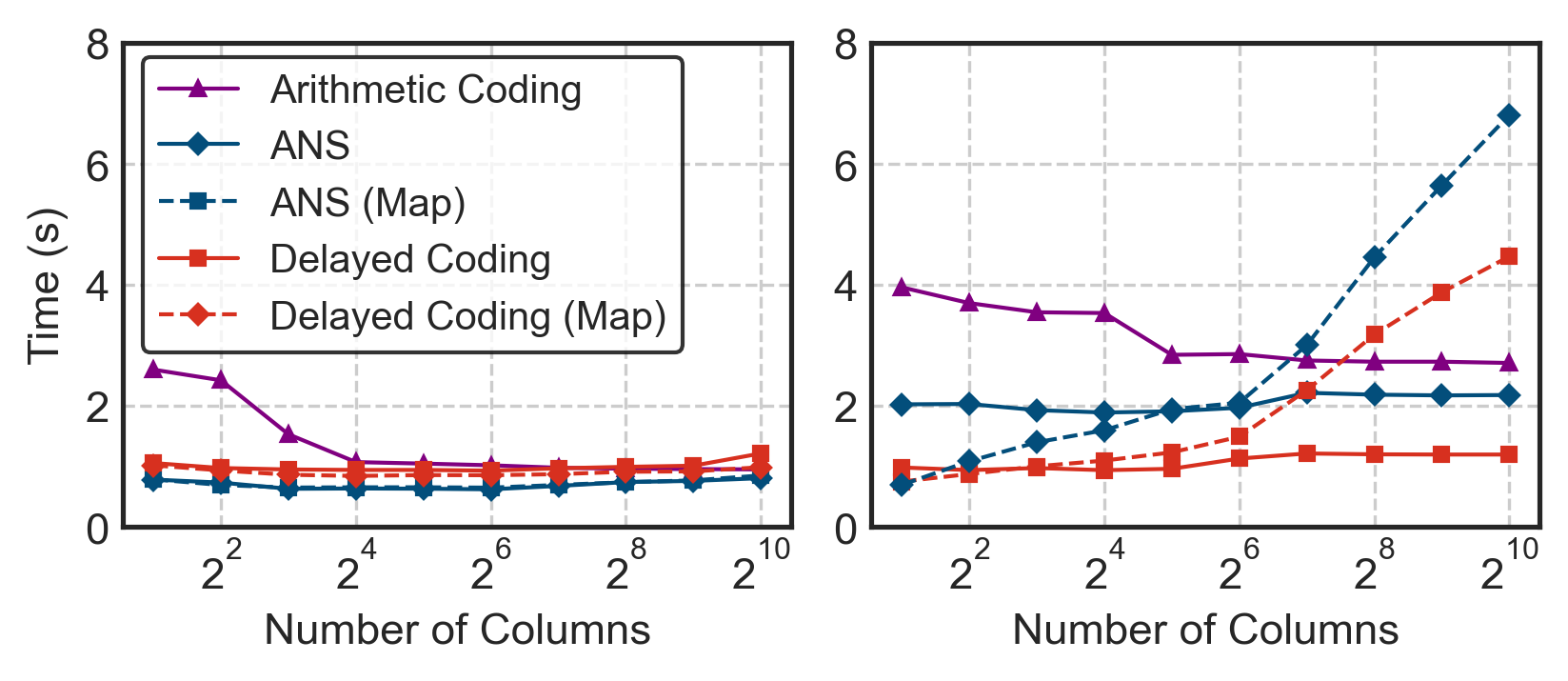}
      {\raggedright
          \hspace{4ex}
          \begin{minipage}{0.05\linewidth}
              \hspace{1ex}
          \end{minipage}
          \begin{minipage}{0.45\linewidth}
                \centering
                \vspace{-6ex}
                \caption*{(a) \textnormal{Compression} }
          \end{minipage}
          \hspace{0ex}
          \begin{minipage}{0.45\linewidth}
                \centering
                \vspace{-6ex}
                \caption*{(b) \textnormal{Decompression} }
          \end{minipage}
      }
      \vspace{-3ex}
      \caption{Entropy Coding Running Time \textnormal{- Vary the \#column of tables and record the processing time.}}
      \vspace{-2ex}
      \label{fig:entropy}
\end{figure}
\subsection{Sensitivity to Sampling Number} 
\label{exp: sensitivity_to_sampling_number}
\work randomly selects a subset of samples for structure learning. We now investigate the sensitivity to the sampling number.
\revision{We use \work w/ Correlation to compress and decompress the whole data sets with the compression granularity being a single tuple (i.e., delayed coding encodes each tuple into a separate compressed block). We keep this tuple-level compression granularity or the remaining experiments unless specified otherwise.}
We select two representative data sets \textit{Bimbo} and \textit{Census} for the analysis. Because the structure learning influences the complexity of the model generated, which further affects the compression speed, we report the duration of each stage within \work: structure learning (Structuring), model generation (Generation), compression, and decompression. 

\cref{fig:sample number} shows the results. We vary the \#samples in structure learning on the x-axis (log-scaled) and record the compression factor and the running time of each stage. 
For \textit{Bimbo} in \cref{fig:sample number}a, the sample number has little effect on the compression factor and running time -- the learner cannot learn many dependencies. Structure learning time increases slightly with sample number; this is expected since more samples need scanning. \textit{Census} in \cref{fig:sample number}b shows a different pattern: the compression factor increases with the \#samples -- the learner finds interesting dependencies and generates more complex models. Therefore, we need more time to generate models. The disparity between the two patterns is due to different column counts: \textit{Census} has 69 columns compared to \textit{Bimbo}'s 12. 
More columns typically indicate more complex dependencies, leading to higher access latency and longer training time. 
Considering the running time and performance, we set the default sample number to $2^{15}$ for \work. 
\revision{{R4.W3, R4.D5}Although considering column correlation may improve the compression factor of \work on a few data sets (e.g., \textit{Yale}) with a small impact on the access latency, we opt to use \work w/o Correlation for the remaining experiments unless otherwise specified.}

\subsection{Entropy Coding Running Time}
\label{exp: entropy_coding}
We evaluate the delayed coding with arithmetic coding and asymmetric numeral systems (ANS) \cite{DBLP:journals/corr/abs-2106-06438-ans}.
We implement arithmetic coding with the integer-based probability representation; 
and use finite-state entropy \cite{ans_impl} as the implementation of the ANS. 
We integrate them into \work and use the same models for distribution estimation. 
In this experiment, we create 64 MB relational data sets with different numbers of columns. 
Each column has a uniform distribution of cardinality 255 with values sampled from ASCII codes.
We vary the column number from 2 to 1024 to record the compression and decompression times of all algorithms. 

In \cref{fig:entropy}, we compare delayed coding, arithmetic coding, and ANS, represented by the solid lines. 
Delayed coding is $2\times$ faster than ANS for the decompression speed, with arithmetic coding being the slowest. This is because delayed coding has a constant-time decoding complexity. 
Arithmetic coding, on the other hand, relies on binary search for code-to-symbol mapping, operating in $O(\log N)$. 
An improvement for ANS involves using an unordered map from codes to symbols to accelerate decoding \cite{DBLP:journals/corr/abs-2106-06438-ans}. 
We then implement such a decoding map for both ANS and delayed coding, denoted by the dotted lines. \cref{fig:entropy} shows that the delayed coding is still faster than ANS. Decompression time for ANS is similar to delayed coding with few columns but slows as column numbers increase due to cache burden. Storing decoding maps increases the cache miss rate from 0.04\% to 0.132\% when columns exceed 108.

\begin{figure}[t!]
    \centering
    \includegraphics[width=0.95\linewidth]{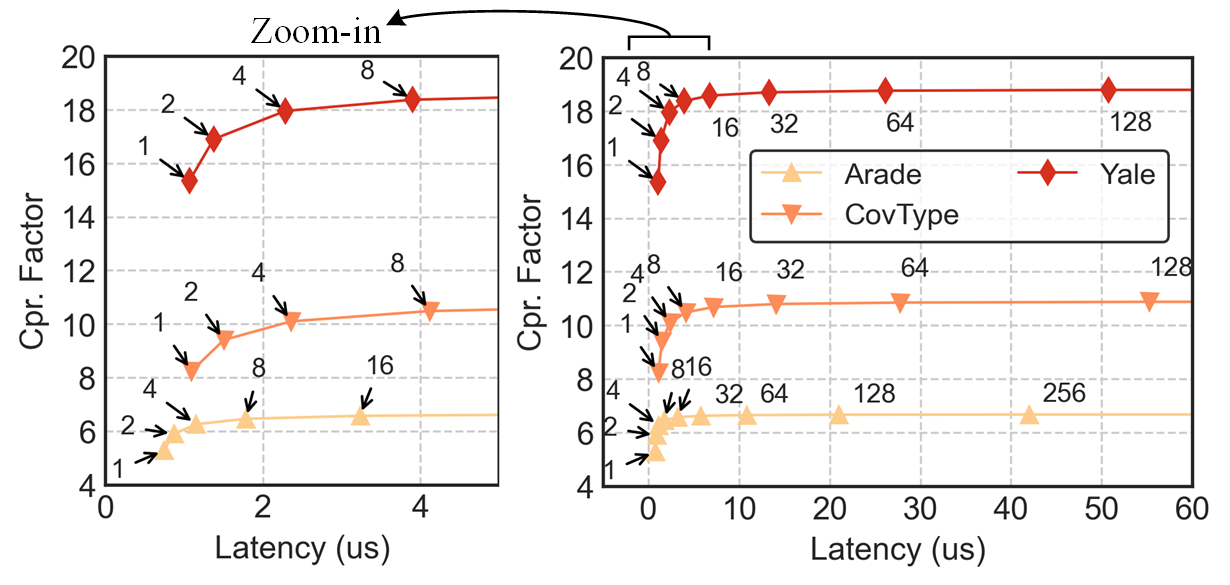}
     {\raggedright
          \hspace{4ex}
          \begin{minipage}{0.05\linewidth}
              \hspace{1ex}
          \end{minipage}
          \begin{minipage}{0.35\linewidth}
                \centering
                \vspace{-5.5ex}
                \caption*{(a) \textnormal{Zoom-in} }
          \end{minipage}
          \hspace{0ex}
          \begin{minipage}{0.52\linewidth}
                \centering
                \vspace{-5.5ex}
                \caption*{(b) \textnormal{Latency vs. Cpr. Factor} }
          \end{minipage}
      }
    \vspace{-2.5ex}
    \caption{Compression Granularity \textnormal{- Vary the block size from 1 to 128 tuples, and write it next to the marker for each trial.}}
    \vspace{-2ex}
    \label{fig:latency}
\end{figure}
\subsection{Sensitivity to Compression Granularity}
\label{exp: sensitivity_to_granularity}
We claim that delayed coding has near-entropy performance with fine compression granularity in \cref{sec: fine_granularity}. 
In this part, we investigate the effect of compression block size in practice. 
We present three data sets for this experiment: \textit{Arade}, \textit{Cover Type}, and \textit{Yale Language}. 
We omit other data sets because they produce similar results. 
In each trial, we vary the block size (\#tuple) to see 
how they affect the compression factor and the random access latency. 
We measure the access latency by repeating the process one million times.
This experiment is conducted in memory.

\cref{fig:latency} shows that the delayed coding has a high compression factor with fine compression granularity. 
For each table, the compression factor reaches a plateau when the compression block size exceeds 8 tuples. 
This indicates a trade-off between compression factor and latency within the 0 to 8 tuple range, 
allowing the users to select their preferred block sizes. 
Since OLTP databases usually prioritize low latency and \work has a high compression factor even at a block size of one tuple, we set this as the default.

\begin{figure}
    \centering
    \includegraphics[width=\linewidth]{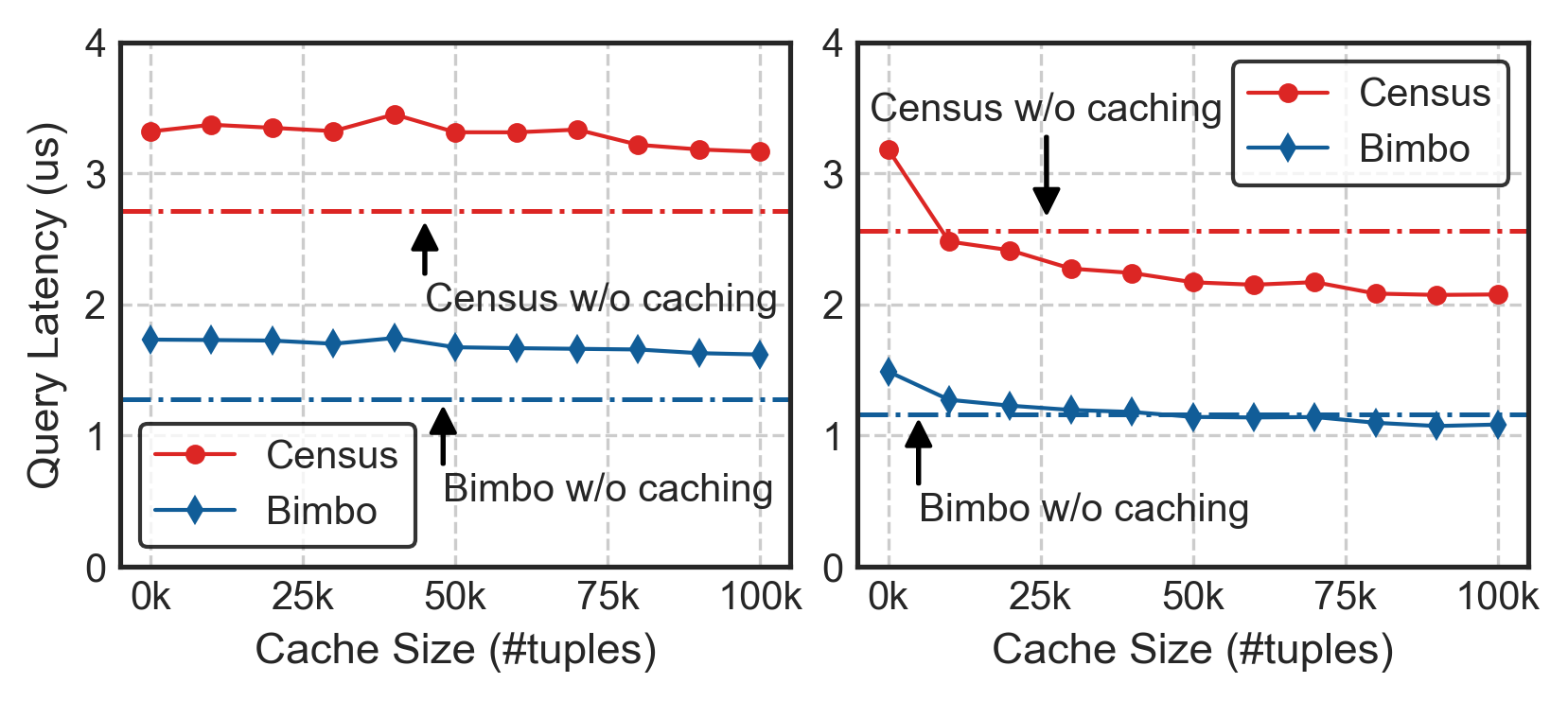}

    {\raggedright
        \hspace{0ex}
        \begin{minipage}{0.05\linewidth}
            \hspace{0.5ex}
            \vspace{-6ex}
        \end{minipage}
        \begin{minipage}{0.45\linewidth}
            \centering
            \vspace{-6ex}
            \caption*{(a) \textnormal{Uniform Distribution} }
        \end{minipage}
        \hspace{0ex}
        \begin{minipage}{0.45\linewidth}
            \centering
            \vspace{-6ex}
            \caption*{(b) \textnormal{Zipfian Distribution} }
        \end{minipage}
    }
    \vspace{-2.5ex}
    \caption{\revision{Effect of a Fast-Path LRU Cache \textnormal{- The workloads are based on YCSB Workload F (read-modify-write). Dashed lines are \work without caching, while the solid lines are \work with fast-path cache enabled.}}}
    \vspace{-2ex}
    \label{fig:access_cache}
\end{figure}

\revision{
\subsection{Fast Path for Tuple Updates}
\label{exp: update_fast_path}
We evaluate a fast path for tuple updates in this section. Specifically, we implemented an LRU write-back cache to buffer the most recently accessed tuples in their decompressed form. Normally, a tuple update involves loading the compressed tuple, decompressing it, modifying the tuple, and re-compressing the updated tuple. With the cache, the workload flow starts by looking up the cache. If the target (decompressed) tuple is already in the cache, we modify the tuple directly. Otherwise, we first decompress the target tuple and insert it into the cache. We evaluate the cache using YCSB workload F (read-modify-write) under Uniform and Zipfian query distributions on data sets \textit{Census} and \textit{Bimbo}. The table initially contains five million tuples, and we execute one million read-modify-write queries on the table.}

\revision{As shown in~\cref{fig:access_cache}, adding the fast-path cache slows down uniformly distributed queries because cache hits are rare, and cache lookup and maintenance bring overhead. For Zipf-distributed queries, having the fast-path cache improves the query performance because as the cache size increases, more and more tuple updates are performed directly in the cache without decompression. We conclude that a fast-path cache can benefit skewed workloads when using \work for compression.}

\section{System Evaluation}
\revision{We integrated \work (w/o Correlation) into Silo \cite{tu2013speedy} and measured the end-to-end performance using the TPC-C benchmark~\cite{tpcc}.} Silo is an OCC-based serializable database designed for excellent performance at scale on large multi-core machines. Silo uses the Masstree\cite{mao2012cache}  for its underlying indexes, and has a very high transaction throughput, achieving more than 1 million txns/s on the standard TPC-C workload in our experiments.

Each record in a table is compressed separately. The compressors under test are Uncompressed, Zstandard, Raman, and \work with the corresponding row-stores named Silo, ZstdDB, RamanDB, and BlitzDB, respectively. To access a record by primary key, Silo walks the index tree using that key to find the compressed record and then decompresses it into a list of attributes. 

\begin{table}[t!]
\centering
\caption{Data Generation Methods \textnormal{- Addresses are generated by ZIP code conventions \cite{us_zip_codes_database}; Phone and district are produced by populating a predefined format with random numbers. }}
%\caption{Data Generation Methods \textnormal{- (1) Names, streets, and miscellaneous data are derived from random sampling of specific sources; (2) Addresses are generated by authentic ZIP code conventions \cite{us_zip_codes_database}; (3) Phone numbers and district information are produced by populating a predefined format with random numbers. }}
\vspace{-2ex}
\begin{tabular}{lcc}
\toprule
\textbf{Column} & \textbf{Method} & \textbf{Source/Format} \\
\midrule
C\_FIRST       & Sampling            & US Baby Names \cite{us_baby_names} \\
C\_STREET      & Sampling            & Open Addresses \cite{real_address} \\
C\_DATA        & Sampling            & City Max Capita \cite{DBLP:conf/sigmod/VogelsgesangHFK18} \\
S\_DATA        & Sampling            & Corporations \cite{DBLP:conf/sigmod/VogelsgesangHFK18}  \\
\midrule
C\_STATE       & Sampling            & List of US States  \\
C\_CITY        & Conditional         & Cities within C\_STATE \\
C\_ZIP         & Conditional         & ZIP Codes within C\_CITY  \\
\midrule
C\_PHONE       & Format-Based& ``(XXX) XXX-XXXX''  \\
S\_DIST        & Format-Based& ``dist-str\#XX\#XX\#XXXX''\\
\bottomrule
\end{tabular}
\label{tbl:data_generation}
\end{table}

According to the TPC-C specification, some columns are filled with random bytes which are incompressible. We substitute these bytes with data that either follows real-world patterns or is sampled from the collected corpus. \cref{tbl:data_generation} details our data generation approach. The compression factors for the new Customer and Stock tables are 3.44 and 5.57, respectively.

\begin{figure*}[t!]
    \centering
    % first column
    \begin{minipage}[b]{0.48\textwidth}
        \centering
        \includegraphics[width=\linewidth]{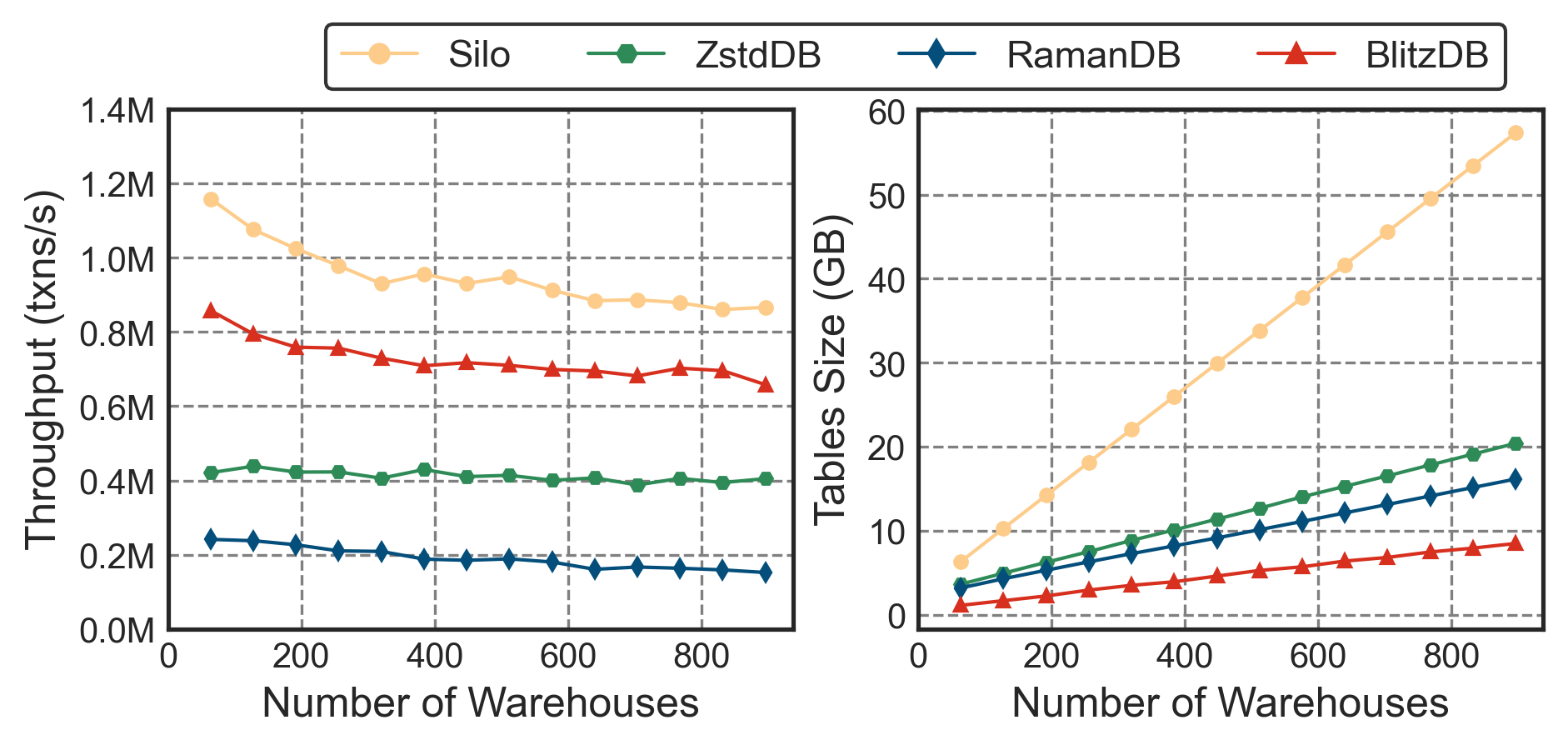}
        {\raggedright
          \begin{minipage}{0.1\linewidth}
          \vspace{-6ex}
          \end{minipage}
          \hspace{3ex}
          \begin{minipage}{0.45\linewidth}
            \vspace{-6ex}
            \caption*{(a) \textnormal{Throughput} }
          \end{minipage}
          \hspace{0ex}
          \begin{minipage}{0.45\linewidth}
            \vspace{-6ex}
            \caption*{(b) \textnormal{DB Size} }
          \end{minipage}
      }
      \vspace{-2.5ex}
      \caption{TPC-C Workload \textnormal{- In each trial, we use 16 threads and each thread executes 1 million transactions. }}
      \label{fig:tpcc_base}
    \end{minipage}
    \hspace{2ex}
    % second column
    \begin{minipage}[b]{0.48\textwidth}
    \centering
      \includegraphics[width=\linewidth]{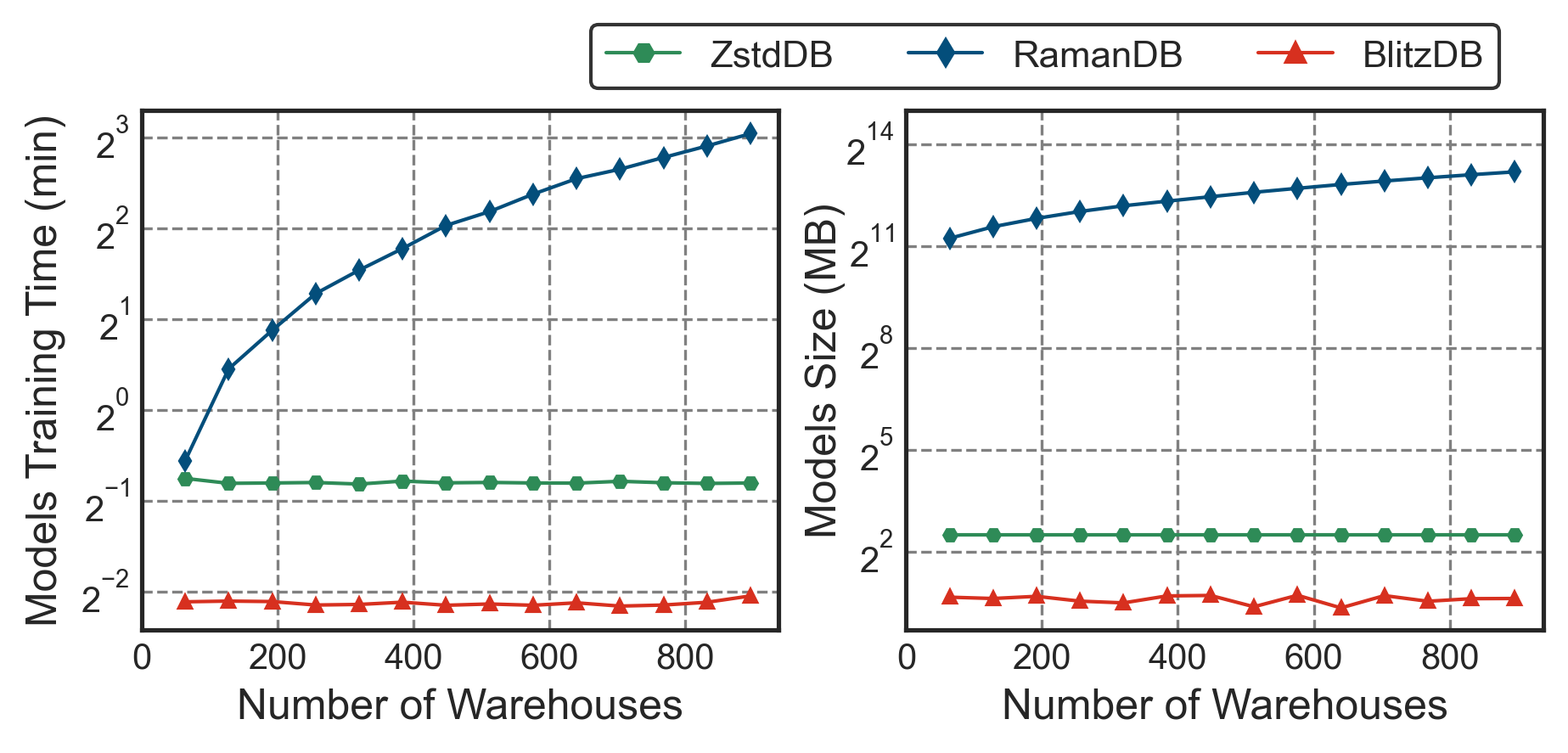}
        {\raggedright
          \begin{minipage}{0.1\linewidth}
              \vspace{-6ex}
          \end{minipage}
          \hspace{3ex}
          \begin{minipage}{0.45\linewidth}
              \vspace{-6ex}
              \caption*{(a) \textnormal{Training Time} }
          \end{minipage}
          \hspace{0ex}
          \begin{minipage}{0.45\linewidth}
              \vspace{-6ex}
              \caption*{(b) \textnormal{Models Size} }
          \end{minipage}
      }
      \vspace{-2.5ex}
      \caption{Models in TPC-C \textnormal{- RamanDB uses full data for training, and the others only sample 16 warehouses for training. }}
      \label{fig:tpcc_model}
    \end{minipage}
    \vspace{-2ex}
\end{figure*}

\subsection{In-Memory Workloads}
In this section, we investigate the performance-space trade-offs of BlitzDB compared to the other baselines when the entire database fits in memory. We vary the number of warehouses in TPC-C from 64 to 896, in increments of 64. In each trial, we use 16 threads. Each database executes 16 million transactions and presents the average throughput results. The training time is measured before the transactions start, while the database size and model size are measured after the transactions. RamanDB uses the entire data set for training, while \work and ZstdDB sample data from 16 warehouses. Because Raman's approach uses a static dictionary, it cannot compress new records. We, therefore, use a buffer (size = 64K tuples) to batch the newly inserted and updated records temporarily. When the buffer is full, we create a new dictionary to compress these buffered records. 
Before adding or using these dictionaries, each thread secures a mutex lock.

As shown in \cref{fig:tpcc_base}, \work compresses the data to 14.8\% of the original (i.e., Silo) with a throughput decrease of around 21\% due to the compression overhead. Such a performance-space trade-off is much more optimized compared to ZstdDB and RamanDB. Moreover, \cref{fig:tpcc_model} shows that \work has the smallest model size and requires orders-of-magnitude shorter time for training compared to the baselines.

\begin{figure}[t!]
  \centering
  \includegraphics[width=\linewidth]{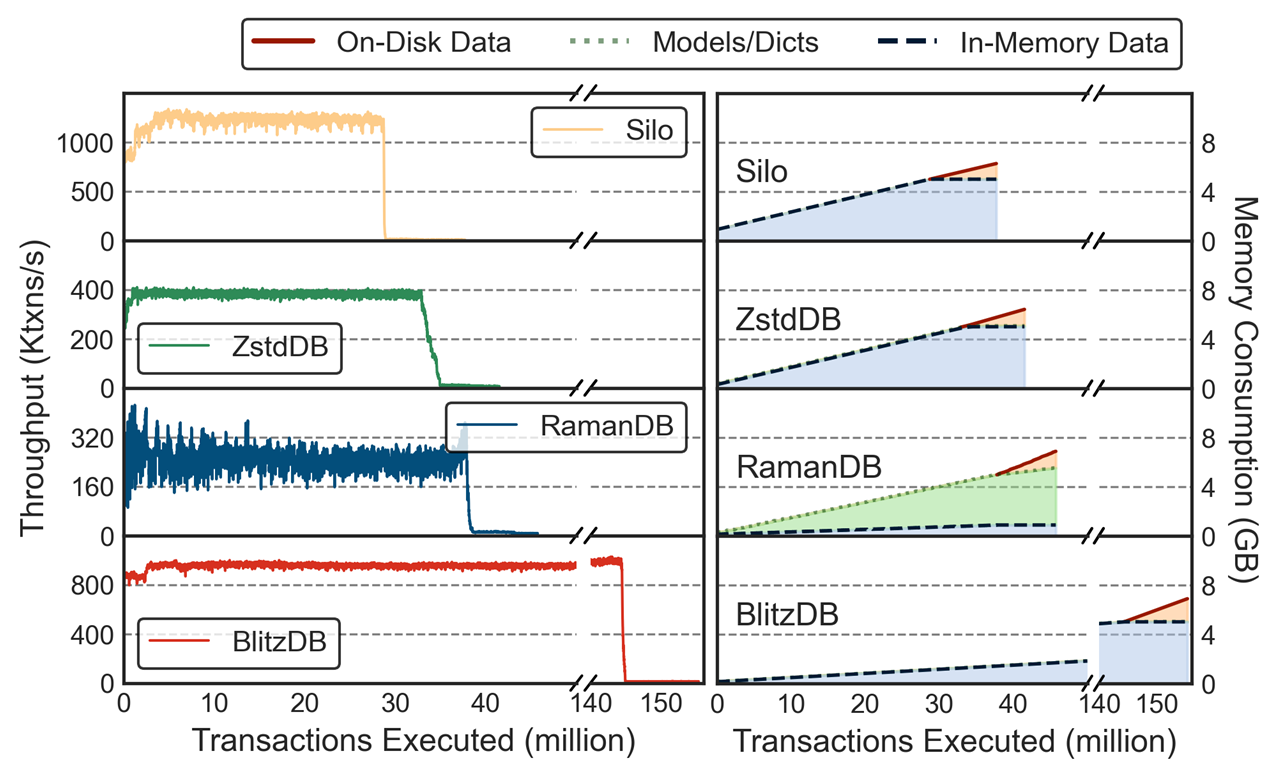}
  {\raggedright 
  
      \begin{minipage}{0.01\linewidth}
          \vspace{-5ex}
      \end{minipage}
      \hspace{3ex}
      \begin{minipage}{0.4\linewidth}
          \vspace{-5ex}
          \caption*{(a) \textnormal{Throughput} }
      \end{minipage}
      \hspace{2ex}
      \begin{minipage}{0.47\linewidth}
          \vspace{-5ex}
          \caption*{(b) \textnormal{Memory Consumption}}
      \end{minipage}
  }
  \vspace{-2ex}
  \caption{TPC-C Large-Than-Memory Workload \textnormal{- Start with 16 warehouses, and each trial runs 20 minutes using 16 threads. }}
  \vspace{-2ex}
  \label{fig:tpcc_ltm}
\end{figure}

\cref{fig:tpcc-scalability} shows the TPC-C throughput as the number of threads grows, with each thread corresponding to a warehouse. Both BlitzDB and Silo demonstrate impressive thread scalability. However, throughput reaches a clear limit after 64 threads, which is particularly noticeable in Raman's approach. This limit can be ascribed to several factors: hyperthreading, the increasing size of the database, shared resources such as the L3 cache, and direct thread contention.

\subsection{Larger-Than-Memory Workloads}
We then evaluate \work under the case when the working set does not fit in physical memory. The tuples are stored on disk with the memory acting as a cache. The memory (i.e., the buffer pool) adopts an LRU replacement policy, and we set the memory limit to 5 GB (excluding the memory occupied by indexes). We start with 16 warehouses (around 1 GB) and execute TPC-C transactions for 20 minutes using 16 threads.

\cref{fig:tpcc_ltm} shows the throughput and memory consumption for the experiments. Note that the x-axis represents the number of executed transactions as did in~\cite{DBLP:conf/sigmod/ZhangAPKMS16}. After 20 minutes of execution, BlitzDB completed $5\times$ as many transactions as Silo. This is because \work not only achieves an exceptional compression factor but also brings moderate compression/decompression overhead to the system. BlitzDB can sustain at high throughput for a longer time because the memory saved by \work allows the database to keep a larger working set in memory.

As a comparison, neither ZstdDB nor RamanDB significantly improves transaction execution. Zstandard suffers from a low compression factor, especially on short tuples. For example, despite using the zstd-dictionary, Zstandard only achieves a compression factor of around 1.3 for the TPC-C table OrderLine. On the other hand, Raman's method is limited by its large dictionary size. It uses a buffer to temporarily hold tuples, compressing them once the buffer is full and then clearing them. This leads to the generation of large compression dictionaries and unstable throughput, with a notable decrease in speed during compressing the buffer tuples.

\section{Related Work}
Lightweight encoding, such as bit-packed, delta, run-length, dictionary, and bit vector encoding, is popular recently \cite{chen2001query, shi2020column, DBLP:conf/sigmod/AbadiMF06ColumnCompression, parquet, abadi2013design, boncz2020fsst, abadi2006integrating}. These are used in column-store databases as they can quickly process large chunks of data using the SIMD technique \cite{jiang2018boosting}. However, delta and run-length encoding methods can be slow when we need to quickly grab just a single tuple/value, as they have to decode an entire data block. Therefore, these lightweight encoding methods are unsuitable for processing data in OLTP databases.

General-purpose block compression methods such as Gzip\cite{DBLP:journals/tit/ZivL77gzip}, Snappy \cite{snappy}, and Zstandard \cite{DBLP:journals/rfc/rfc8878zstd} effectively save disk space by using a sliding window technique to identify word repetitions, thus minimizing data transfer between disk and memory \cite{parquet}. However, the static dictionary makes them \textit{inflexible} for insert/update scenarios. Zstandard provides a special mode for small files. It improves compression by training on all small file, generating a dictionary. This dictionary is required to be loaded before compression and decompression. However, this mode is less effective at compressing files slightly different from the prior data.

Semantic compression needs to estimate probability distributions for each column in a relational table. Babu et al. proposed the first lossy semantic compression method, SPARTAN, for table compression \cite{DBLP:conf/sigmod/BabuGR01}. Subsequently, Gao et al. introduced Squish, which uses a Bayesian network and arithmetic coding \cite{DBLP:conf/kdd/GaoP16}. Later, DeepSqueeze was conceived, using auto-encoders \cite{DBLP:conf/sigmod/IlkhechiCGMFSC20}. However, being lossy, these techniques are suited for data archiving and less so for low-latency transaction processing. For example, Squish sorts each table column to make delta encoding more efficient, but slowing compression. DeepSqueeze does not support columns with high cardinality\footnote{DeepSqueeze uses one-hot encoding for each column in its network architecture, and high cardinality introduces numerous parameters in the fully connected layer~\cite{DBLP:conf/sigmod/IlkhechiCGMFSC20}.}. 
In contrast, \work supports all common column types in databases and provides a faster compression speed. \revision{
Also, DeepSqueeze uses deep learning techniques for structure learning~\cite{DBLP:conf/sigmod/IlkhechiCGMFSC20}. However, this approach lacks explainability and has a slow inference speed. \work, therefore, uses the Bayesian network to capture column correlations for its simplicity. The effectiveness of the Bayesian network approach has been proved in~\cite{DBLP:conf/kdd/GaoP16}.}

\begin{figure}[t!]
    \centering
    \includegraphics[width=0.65\linewidth]{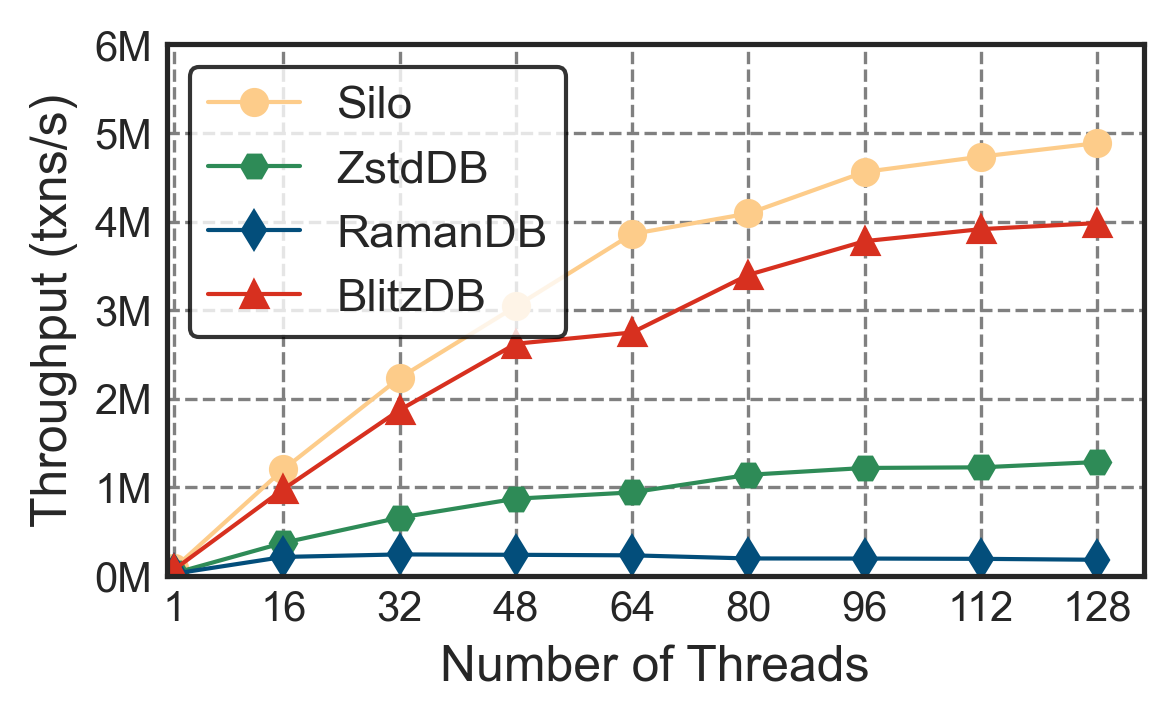}
    \vspace{-3ex}
    \caption{Scalability of Compression \textnormal{- Vary the thread number from 1 to 128, and each trail runs 1 minute.}}
    \label{fig:tpcc-scalability}
    \vspace{-2.5ex}
\end{figure}

\section{Conclusions}
We introduce \work, a high-speed semantic compressor for OLTP databases. We first propose novel semantic models that support fast inferences and dynamic value sets for both discrete and continuous data types; we then introduce a new entropy encoding algorithm, called delayed coding, that achieves significant improvement in the decoding speed compared to modern arithmetic coding implementations. \work has high compression factors and fast decompression speed. We integrate \work into an in-memory OLTP database, Silo. The TPC-C benchmark shows that, for data sets larger than the available physical memory, \work can help the database sustain a high throughput and execute four times more transactions before the I/O overhead dominates.

% \begin{acks}
% Huanchen Zhang and Yihan Gao are the corresponding authors.
% \end{acks}

\clearpage
\balance
\bibliographystyle{ACM-Reference-Format}
\bibliography{base}

\clearpage
\appendix
\section{Arithmetic Coding}
In this section, we give the pseudo-code of arithmetic coding based on integer-based probability representation, as shown in \cref{alg:arithmetic}. All algorithms shown in the appendix are in C++ style. Input of function \texttt{Encode} is a sequence of probability intervals, we first compute the product of these intervals, then find a suitable integer $M$ to represent the product result. Note that some codes are generated during the product computing process, this is because we need an early bits emission technique to avoid precision underflow. 

\begin{algorithm}
\caption{Encoding of Arithmetic Coding}
\label{alg:arithmetic}
\DontPrintSemicolon
\SetKwProg{Fn}{Function}{:}{end}
\SetKwFunction{ArithmeticEncoding}{Encode}

\SetKwData{Range}{$range$}
\SetKwData{Lbase}{$L_{32}$}
\SetKwData{Rbase}{$R_{32}$}
\SetKwData{L}{$L$}
\SetKwData{R}{$R$}
\SetKwFunction{GetPIProduct}{GetPIProduct}

\Fn{\ArithmeticEncoding{$[L_1, r_1], \cdots, [L_s, r_s]$}}{
    $codes$ $\leftarrow$ $\varnothing$ \;
    \L $\leftarrow$ 0 \;
    \R $\leftarrow$ 65536 \;
    
    \For{$i \leftarrow 1$ \KwTo $s$}{
        $code$, $[L, R]$ $\leftarrow$ \GetPIProduct{$[L ,R]$, $[L_i, R_i]$} \;
        $codes \leftarrow codes + code$ \; 
    }
    
    Find smallest $k$ such that $\exists M, [2^{-k} M, 2^{-k} (M + 1)] \subseteq [L, R]$, \;
    \KwRet $codes + M$ \;
}

\Fn{\GetPIProduct{$[L_a, R_a]$, $[l_b, r_b]$}}{
    \Range $\leftarrow$ $R_a - L_a$ \;
    \Lbase $\leftarrow$ ($L_a \ShiftLeft 16 $) + $\Range \cdot l_b$ \;
    \Rbase $\leftarrow$ ($L_a \ShiftLeft 16 $) + $\Range \cdot r_b$ \;
    \L $\leftarrow$ $\Lbase \ShiftRight 16$ \;
    \R $\leftarrow$ $\Rbase \ShiftRight 16$ \;
    
    \tcc{Underflow Check}
    \uIf{\R $>$ \L + 1}{
        \KwRet $\varnothing$, $[\L, \R]$
    } \uElseIf{\R $=$ \L + 1}{
        \tcc{Truncate $[L, R]$ containing 32768 if necessary. }
        \eIf{$(\R \ShiftLeft 16) - \Lbase \geq \Rbase - (\R \ShiftLeft 16)$}{
            \KwRet $\L \ \&\  \texttt{0xffff}$, $[\Lbase\ \&\ \texttt{0xffff}, 65536]$\;
        }{
            \KwRet $\R \ \&\  \texttt{0xffff}$, $[0, \Rbase\ \&\ \texttt{0xffff}]$\;
        }
    } \Else{
        \tcc{Underflow happens, i.e., $l = r$.}
        \KwRet $\L \ \&\  \texttt{0xffff}$, $[\Lbase\ \&\ \texttt{0xffff}, \Rbase\ \&\ \texttt{0xffff}]$\;
    }
}
\end{algorithm}

\textbf{Precision Underflow is Tackled By Early Bits Emission. } In practice, there could be a large number of probability intervals for one record, so the product can easily exceed the precision limit, i.e. $L' \geq R'$ (actually only $L' = R'$ is possible), where $[L', R']$ is the product result. Early bits emission are leveraged in arithmetic coding with integer-based probability intervals. Suppose underflow happens, i.e., $L' = R'$, according to the early bits emission, we can emit the 16-bit $L_{32} \ShiftRight 16$ and let
\begin{equation*}
\begin{aligned}
    L' = L_{32}\ \&\ \texttt{0xffff}, 
    R' = R_{32}\ \&\ \texttt{0xffff}.
\end{aligned}
\end{equation*}
This technique is very interesting, underflow happens if the first 16 bits of $L_{32}$ and $R_{32}$ are the same (we use 32 bits to temporarily represent the product of probability intervals), which means that the next two bytes representing the data have been determined. Therefore, we can directly output the high 16 bits and update the product result in arithmetic coding. 

\textbf{Bits Prefetch. } With integer-based probability representation, decoding of arithmetic coding is similar to the original version \cite{MacKay2003InformationTheory}, except that we know the next symbol can be determined by reading at most 16 bits. Thus, we can read 16 bits each time since the extra bits read will not disturb decoding. In this way, we do not need to check whether the current information is enough to decode the next symbol.

\section{The Constant Time Complexity of Inv-Translate}
Every probability vector  $\pi_1, \cdots, \pi_N$, can be expressed as an equiprobable mixture of $N$ two-point distributions. That is, there are $N$ pairs of integers $(\alpha_1, \beta_1)$, $\cdots$, $(\alpha_N, \beta_N)$ and $N$ probabilities $w_1, \cdots, w_N$ such that
\begin{equation*}
    \pi_i = 1/N \cdot \sum_{j = 1}^N (w_j \mathds{1}_{\{\alpha_j = i\}} + (1 - w_j) \mathds{1}_{\{\beta_j = i\}}) = 1/N \cdot \sum_{j = 1}^N Y^{(j)}_i
\end{equation*}
for $1 \leq i \leq N$, where $Y^{(1)}$,~$\cdots$,~$Y^{(N)}$ are two-point distributions.

This can be shown by induction. It is true when $N = 1$. Assuming that it is true for $N < k$, we can show it is true for $N = k$ as follows. Choose the minimal $\pi_i$. Since it is at most equal to $1/N$, we can take $\alpha_1$ equal to the index of this minimum and set $w_1$ equal to $N\pi_{\alpha_1}$. Then choose the index $\beta_1$ which corresponds to the largest $\pi_i$. This defines the first function $Y^{(1)}$. Note that we used the fact that $(1 - w_1)/ N \leq \pi_{\beta_1}$ because $1 / N \leq \pi_{\beta_1}$. The other $N-1$ functions have to be constructed from the leftover probabilities 
\begin{equation*}
    \pi_1, \cdots, \pi_{\alpha_1} - \pi_{\alpha_1}, \cdots, \pi_{\beta_1} - (1 - w_1) / N, \cdots, \pi_N
\end{equation*}
which, after deletion of the $\alpha_1$-th entry, is easily seen to be a vector of $N-1$ non-negative numbers summing to $(N-1)/N$. For such a vector, the left $Y^{(j)}$s can be found by our induction hypothesis.

\section{Delayed Coding}
In this section, we give the pseudo-code of delayed coding in detail.
\subsection{Encoding Procedure}
The pseudo-code of encoding is shown in \cref{alg:delayed_coding_encoding}. Encoding of delayed coding has two stages: (1) Planning; (2) Filling. 
\begin{algorithm}
\caption{Encoding Procedure of Delayed Coding}
\label{alg:delayed_coding_encoding}
\DontPrintSemicolon{}
\SetKwProg{Fn}{Function}{:}{}
\SetKwFunction{Encode}{Encode}
\Fn{\Encode{$[L_1, R_1], \cdots, [L_s, R_s]$}}{
    \tcc{\textbf{1. Planning.}} 
    $size \leftarrow 1$ \;
    $isVirtual \leftarrow false$ \;
    \While{$i \leftarrow 1$ \KwTo $s$}{
        Mark the interval $[L_i, R_i]$ with $isVirtual$. \;
        $isVirtual \leftarrow false$ \;
        
        $k \leftarrow R_i - L_i$ \;
        $size \leftarrow size \cdot k$ \;
        \If{$size \geq 2^{16}$}{
            $isVirtual \leftarrow true$ \;
            $size \leftarrow size \ShiftRight 16$ 
        }
    }
    
    \tcc{\textbf{2. Filling.}} 
    $data \leftarrow 0$ \;
    $bitStream \leftarrow \varnothing$ \;
    \For{$i \leftarrow s - 1$ \KwTo $0$}{
        $k \leftarrow R_i - L_i$ \;
        $a \leftarrow data \bmod k$ \;
        $data \leftarrow data / k$ \;
        $16bits \leftarrow L_i + a$ \;
        
        \eIf{$[L_i, R_i]$ is virtual}{
            $data \leftarrow (data \ShiftLeft 16) + 16bits$ 
        }{
            $bitStream \leftarrow 16bits + bitStream$ 
        }
    }
    \KwRet $bitStream$ \;
}
\end{algorithm}

\textbf{Planning. } The planning stage is to determine how many and which intervals can be virtual. The planning stage is necessary since we need to fill virtual bits from end to start, as proved in Section 5.3. Note that even if a certain probability interval is claimed to be virtual, we still need to calculate whether it can contribute more space to the virtual bits. In other words, every virtual bit is put to good use. Every time the virtual bits number is larger than 40 bits, we state that one more probability interval can be virtual, as shown in lines 9-11. 

\textbf{Filling. } The filling stage is to fill the virtual bits with information we have collected. According to Section 5.3, we compute each $k_i$ and $a_i$($k$, $j$ in pseudo-code) for each probability interval. If an interval is virtual, we add it to $d_i$ ($data$ in pseudo-code); otherwise, we add it to $bitStream$. Note that in the filling phase, we traverse each probability interval from back to front.

\subsection{Decoding Procedure}
The pseudo-code of decoding is shown in \cref{alg:delayed_coding_decoding}.
Decoding is the inverse process of encoding in delayed encoding. Each time we read 16 bits from the bit stream or virtual bits (depending on the information we collect $data$). The function \texttt{Inv-Translate} is then called to get the symbol corresponding to the 16 bits of the input, and additional information $a/k$. We update the virtual bits with additional information and check if the next virtual probability interval can be obtained. 

The function \texttt{Inv-Translate} is inspired by the alias method, which is a constant time sampling algorithm from a categorical distribution \cite{kronmal1979alias}. Alias method needs a simple prepossessing step, whose time complexity is $O(n)$, and the intuition is that we can partition the probability interval $[0,1)$ into a series of buckets such that when we pick a random value (16bits) in the range, it ends up in some bucket with probability equal to the size of the bucket. 

\textbf{Alias Method. } Let us suppose the finite alphabet has $n$ characters, with weight $k_0, \cdots, k_{n-1}$. Note that we use integer-based probability with a fixed denominator, so the length of probability for each character $[l_i, r_i]$ must be the form $k_i \cdot 2^{-16}$. Then the sum of length should be equal to one, i.e., $\sum_{i = 0}^{n-1} k_i = 65536$. When given two bytes as input, we want to find the corresponding character and the value of $j$ very efficiently. This can be done through a simple trick: let $m$ be such that $2^{m - 1} < n \leq 2^m = M$, we set up $2M$ numbers $a_0, b_0, \cdots, a_{M-1}, b_{M - 1}$ such that:
\begin{enumerate}
    \item $a_i + b_i = 2^{16 - m}$;
    \item each $a_i, b_i$ is associated to one of the $n$ characters $u_i \in [n], v_i \in [n]$;
    \item $\forall j, \sum_{u_i = j} a_i + \sum_{v_i = j} b_i = k_j$, the total sum of the number associated with a character is equal to its weight.
\end{enumerate}
This decomposition is possible for any discrete distribution with a finite number of outcomes. The input 16 bits can be regarded as a value between 0 and 65536, and the first $m$ bit represents the index of a bucket. Since each bucket contains two characters, we need to find the correct character by comparing the last 16 - $m$ bits with the character boundary $a_P$. Once the character is identified, we can output it along with extra information. The function \texttt{Translate} has a constant time complexity. The latter two terms of line 7 and line 10 can be calculated in advance, making this function faster. 

\begin{algorithm}
\caption{Decoding Procedure of Delayed Coding}
\label{alg:delayed_coding_decoding}
\DontPrintSemicolon{}
\SetKwProg{Fn}{Function}{:}{end}
\SetKwFunction{Decode}{Decode}
\Fn{\Decode{bitStream, $record$}}{
    $data \leftarrow$ 0 \;
    $size \leftarrow$ 1 \;         
    
    \While{$record$ has unfilled attributes}{
        $16bits \leftarrow \text{next 16-bit from bitStream}$ \;
        $value, a, k \leftarrow \texttt{Inv-Translate}(16bits)$ \; 
        \tcc{\texttt{Inv-Translate} can be a map simply.}
        $data \leftarrow data \cdot k + a$ \;
        $size \leftarrow size \cdot k$ \;
        Fill $record$ with $value$. \;
        
        \If{$size \geq 2^{16}$}{
            $interval \leftarrow data \ \&\  \texttt{0xffff}$ \;
            $bitStream \leftarrow interval + bitStream$ \;
            $data \leftarrow data \ShiftRight 16$ \;
            $size \leftarrow size \ShiftRight 16$ \;
        }
    }
    
    \KwRet $rd$ \;
}
\end{algorithm}

\begin{algorithm}
\caption{Inv-Translate of \work{}}
\DontPrintSemicolon{}
\SetKwProg{Fn}{Function}{:}{end}
\SetKwFunction{Decode}{Decode}
\SetKwFunction{InvTranslate}{Inv-Translate}
Find $m, \{a_{M}\}, \{b_{M}\}, \{u_{M}\}, \{v_{M}\}$ by performing decomposition on $\{k_n\}$. \;
\Fn{\InvTranslate{16bits}}{
    $P \leftarrow 16bits \ShiftRight (16 - m)$ \;
    $Q \leftarrow 16bits \ \&\  (2^{16-m} - 1)$ \;
    \eIf{$Q < a_P$}{
        $c \leftarrow u_P$ \;
        $j \leftarrow 16bits - \sum_{z < P} a_z \mathds{1}_{\{u_z = c\}} - \sum_{z < P} b_z \mathds{1}_{\{v_z = c\}}$
    }{
        $c \leftarrow v_P$ \;
        $j \leftarrow 16bits - \sum_{z \leq P} a_z \mathds{1}_{\{u_z = c\}} - \sum_{z < P} b_z \mathds{1}_{\{v_z = c\}}$
    }
    
    \KwRet $c, j, k_{c}$ \;
}
\end{algorithm}

\textbf{Decomposition. }It remains to show that such decomposition is indeed possible, which can be obtained through the following procedure: 
\begin{enumerate}
    \item Initially, let $S = \{i: k_i < 2^{16 - m}\}$, $L = \{i: k_i \geq 2^{16 -m}\}$.
    
    \item If $S \neq \varnothing$, choose $k_s \in S, k_l \in L$, associate $a_i, b_i$ to $k_s$ and $2^{16 - m} - k_s$, respectively, add $k_l - (2^{16 - m} - k_s)$ to $S$ or $L$; accordingly.
    
    \item Otherwise, choose $k_l \in L$, associate $a_i, b_i$ to $0, 2^{16 - m}$ respectively, add $k_l - 2^{16 - m}$ to $S$ or $L$ accordingly.
    
    \item If $S \neq \varnothing$ or $L \neq \varnothing$, return to (2).
\end{enumerate}
It can be proved via induction that at the end of step (3), $|S| + |L| < |S'| + |L'|$
always holds, where $S', L'$ are the sets in the previous loop after step (3), so the correctness of the procedure can be guaranteed.

\textbf{Discussion: } The Alias method tells us that we can sample from a discrete distribution with complexity $O(1)$. This process has two parts: (1) Get a random value; (2) Find the corresponding symbol in the distribution. In decompression of delayed coding, we read the ``random value'' from the bit stream and then generate the corresponding symbols. Such symbols are prepared and stored in compression deliberately. 

% \onecolumn
% \begin{multicols}{2}

\section{Uniqueness and Efficiency}
\subsection{Uniqueness of Delayed Coding }
We then show that delayed code is uniquely decodable. We first prove such a virtual bits input can be constructed. Recall that we use 16 bits to represent probability interval and collect extra bits in each probability interval as virtual bits. Suppose there are some probability intervals $[L_i, R_i]$, where $k_i = R_i - L_i$, $i = 1, \cdots, n$, and $\prod_{i = 1}^n k_i \geq 2^{16}$. We want to store a probability interval (an integer of 16 bits) $D$ with virtual bits provided by its former intervals, where $D \in [1, 2^{16})$. Note that for each $[L_i, R_i)$, we can store integer $a_{i}$ if and only if $0 \leq a_{i} / k_i < 1$. Thus, let $d_{n+1} = D$, and
\begin{equation}
\begin{aligned}
    a_{i} = d_{i}\ \%\ k_{i},\ \ \ d_{i} = \lfloor d_{i+1} / k_{i+1} \rfloor,
\end{aligned}
\label{encode_formula}
\end{equation}
for $i = 1, \cdots, n$. Then the following decomposition holds:
\begin{equation*}
\begin{aligned}
    D = d_{n+1} = d_{n} \cdot k_n + a_{n} = \cdots = \sum_{i = 1}^n a_i \prod_{j = i + 1}^n k_j \\
\end{aligned}
\end{equation*}
which means that information $D$ are divided into $\{a_n\}$ and stored by these probability intervals. Then, the code corresponding to probability interval $[L_i, R_i)$ is $L_i + a_{i}$. Intuitively, $d_i$ denotes the information we have gained from the first $i-1$ probability intervals. For example, $d_{n+1} = D$, $d_1 = 0$ since 
\begin{equation*}
    0 \leq d_{1} = \lfloor d_{2} / k_{1} \rfloor \leq \cdots \leq \lfloor d_{n+1} / \prod_{i = 1}^n k_i \rfloor = \lfloor D / \prod_{i = 1}^n k_i \rfloor = 0.
\end{equation*}
Moreover, we can get $D$ from $\{a_n\}$, notice that
\begin{equation}
    d_{i+1} = d_{i} \cdot k_{i} + a_{i}.
    \label{decode_formula}
\end{equation}
The virtual bits can be updated according to \cref{decode_formula} since $d_{i}$ is known and $a_{i}$, $k_{i}$ are given by translating the $i$th probability interval.

Since virtual bits input has been proven to be correct, it suffices to show the obtained code $L_i + a_{i}, i = 1, \cdots, n$ is uniquely decodable because its uniqueness is equivalent to a simple concatenation of probability interval codes. Recall each outcome is assigned a disjoint probability interval. Therefore, any number in the interval would be a unique identifier, note that 
\begin{equation*}
    L_i + a_{i} \in [L_i, R_i)
\end{equation*}
for $i = 1, \cdots, n$ since $a_{i} = d_{i}\ \%\ k_{i}$, then $L_i + a_{i}$ is obvious a unique representation. Also, since each codeword is 16-bit, it is a prefix code. Prefix code is always uniquely decodable, so is delayed code.

\subsection{Efficiency of Delayed Coding } 
Suppose $[L_1, R_1), \cdots, [L_n, R_n)$ are $n$ intervals, and $L_i - R_i$ are random variables such that 
\begin{equation*}
    P(R_i - L_i = x) = \frac{x}{\sum_{j = 1}^{65536} j}, x \in [0, 65536]; \text{otherwise } 0
\end{equation*}
Let $\mu = \mathbb{E}(R_i - L_i)$. Suppose delayed coding generates one virtual 16-bit once the virtual bits number is larger than $\alpha$, $\alpha \geq 16$. Delayed coding encodes every $m$ intervals as a block, where $n+1 \geq m >> \mu$. 

Compared to entropy, delayed coding may use more bits for two reasons: (1) integer division; and (2) unused virtual bits when coding ends. 

\textbf{Integer Division Loss. }Let $D_i$ be the $i^{th}$ time to generate one virtual 16-bit, and the index of current interval is $I_i$. We have $D_i \geq 2^\alpha$, and $E(I_1) = \alpha / (1 - \log_2(\mu))$. Also, let $\prod_i = \prod_{j = I_i}^{n} (R_j - L_j)$. 
Notice that 
\begin{equation*}
\begin{aligned}
    D_i / 2^{16} &= \lfloor D_i / 2^{16} \rfloor + r_i / 2^{16},
\end{aligned}
\end{equation*}
where $r_i = D_i \% 2^{16}$. According to entropy, the total number of virtual bits we can get for use from interval $I_i$ to $n$ is $16 + \log_2(D_i / 2^{16}\cdot\prod_i)$. But due to the existence of integer division, we only have
\begin{equation*}
    16 + \log_2(\lfloor \frac{D_i}{2^{16}} \rfloor\cdot\prod_i) = 16 + \log_2((\frac{D_i}{2^{16}} - \frac{r_i}{2^{16}})\cdot\prod_i)
\end{equation*}
bits we can use for encoding the $I_i$ to $n$ intervals. Thus, bits loss due to the $i^{th}$ virtual bytes generation is 
\begin{equation*}
    \log_2(\frac{D_i}{D_i - r_i}).
\end{equation*}
Note that $D_i \geq 2^\alpha$ and $\alpha \geq 16$, then we can get the total bits loss due to integer division,
\begin{equation*}
    \sum_i \log_2(\frac{D_i}{D_i - r_i}) \leq \sum_i \log_2(\frac{2^{\alpha}}{2^{\alpha} - r_i}).
\end{equation*}
Notice there are at most $log_2 \mu^n = n\log_2 \mu$ redundant bits in these intervals, so the number of virtual bytes generation is less than $(n \log_2 \mu - (\alpha - 16))/16$, also we have $r_i = D_i \% 2^{16} < 2^{16}$, therefore
\begin{equation*}
    \sum_i \log_2(\frac{2^{\alpha}}{2^{\alpha} - r_i}) \leq  \frac{n \log_2 \mu - (\alpha - 16)}{16} \cdot \log_2 \frac{2^\alpha}{2^{\alpha} - 65535}.
\end{equation*}

\textbf{Unused Bits Loss. } On the other hand, we consider the unused virtual bits when coding ends. Since there are $n$ intervals and $m$ intervals are encoded as a block, we have $\lfloor n / m \rfloor + 1$ blocks. Therefore, total bits loss due to unused virtual bits in the end is 
\begin{equation*}
    \sum_{k = 1}^{\lfloor n / m \rfloor + 1} \log_2 D_{\infty}^k \leq (\lfloor n / m \rfloor + 1) \cdot \alpha,
\end{equation*}
where $D_{\infty}^k$ denotes the number of virtual bits when $k^{th}$ block ends, we have $2^{\alpha - 16} \leq D_{\infty}^{k} < 2^{\alpha}$. 

\textbf{Length of Delayed Codes. } Thus, the length of delayed codes is 
\begin{equation*}
    L^{n, \alpha, m} \leq n \cdot C + (\lfloor n / m \rfloor + 1) \cdot \alpha + \frac{n \log_2 \mu}{16} \cdot \log_2 \frac{2^\alpha}{2^{\alpha} - 65535}.
\end{equation*}
where $C$ is the entropy of an interval with the given distribution. In fact, $\mu \approx 43690.33$. Simplify it, we get 
\begin{equation}
    L^{n, \alpha, m} \leq n \cdot C + (\lfloor n / m \rfloor + 1) \cdot \alpha + n\cdot \log_2 \frac{2^\alpha}{2^{\alpha} - 65535}.
\end{equation}
Let $m = n + 1$ and $\alpha = 24$, we get the \textbf{archive mode} setting of delayed coding and 
\begin{equation*}
    L^{n, 24, n+1} \leq n \cdot C + 24 + \cdot 10^{-8} \cdot n.
\end{equation*}
Since an interval size is $\log_2(\mu) >> 10^{-8}$ byte in average, 
Let $\alpha = 16$, we get the \textbf{random access mode} setting of delayed coding and
\begin{equation*}
    L^{n, 16, m} \leq n \cdot C + 16 \cdot (\lfloor n / m \rfloor + 1 + n).
\end{equation*}

\section{Integrated Semantic Models}
In this section, we introduce two integrated semantic models: the JSON node model and the time-series model.

\subsection{JSON Node Model}
\work can compress the collection of JSON objects. Each object follows a JSON schema, where optional nodes and multi-type nodes are allowed \cite{DBLP:conf/www/PezoaRSUV16JSONSchema}, for example, $A = \{\text{Name}: ``John", \text{Age}: 18, \text{Job}: ``student"\}$, and $B = \{\text{Name}: ``Mary", \text{Age}: ``Eighteen" \}$, where the absence of the \textit{Job} node in $B$ and the differing types of \textit{Age} in them make it challenging to compress.

A JSON node model consists of metadata, attributes, and sub-models, as shown in \cref{json}. The metadata model consists of two categorical models: the existence model and the type model. The former validates node presence using two unique values: existed or not; while the latter verifies node type, and has four possible outputs. The multi-type node has more than one attribute model. Sub-models handle arrays or objects with a vector of pointers to other JSON node models. Together, all JSON node models form the JSON tree model for nested calling. For example, the root node ``Person'' is an object with three children. The optional ``Job'' node has an additional existence model, while the multi-type node ``Age'' requires a type model. Given a node, \work verifies its existence, determines its type, and then processes it using the appropriate attribute models. If the node is an array or object, other JSON node models are invoked to handle it and generate intervals.

\begin{figure}[t!]
  \centering
  \includegraphics[width=\linewidth]{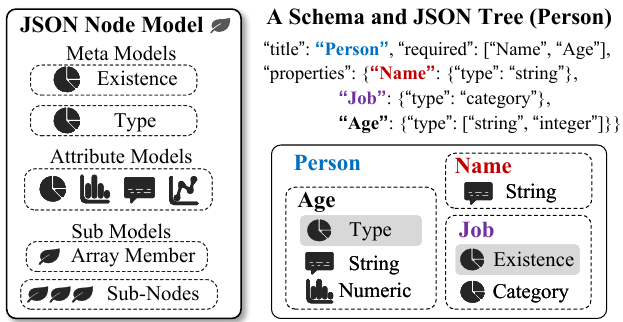}
  \caption{JSON Node Model - \textnormal{The JSON node model is shown on the left. Node ``Job'' is optional, and node ``Age'' has multiple types.}}
  \label{json}
\end{figure}

\subsection{Time-series/Markov Model}
Both the time-series and Markov models capture the first-order transition property in a continuous or discrete column. The time-series model uses the Autoregressive-moving-average (ARMA) technique to decompose a series of numeric values into residuals and a regression model \cite{box2015timeseries}. Compressing these residuals instead of the original values improves performance since the residuals have fewer outliers and more symmetry distribution. Similarly, the Markov model is for a Markov chain categorical column. It has multiple categorical models, each corresponding to a unique value/state. The current value determines which categorical model is used to compress the next one, providing larger intervals for delayed coding. However, the time-series/Markov model in \work cannot be used with the record index.

\section{Experiments}
In this section, we evaluate \work with Squish \cite{DBLP:conf/kdd/GaoP16}, Gzip \cite{DBLP:journals/tit/ZivL77gzip} and Zstandard (level-9) for the table archive task. We also give the evaluation of the times-series and JSON model.
\subsection{Archive Compression}
\begin{figure*}[t!]
  \centering
  \includegraphics[width=\linewidth]{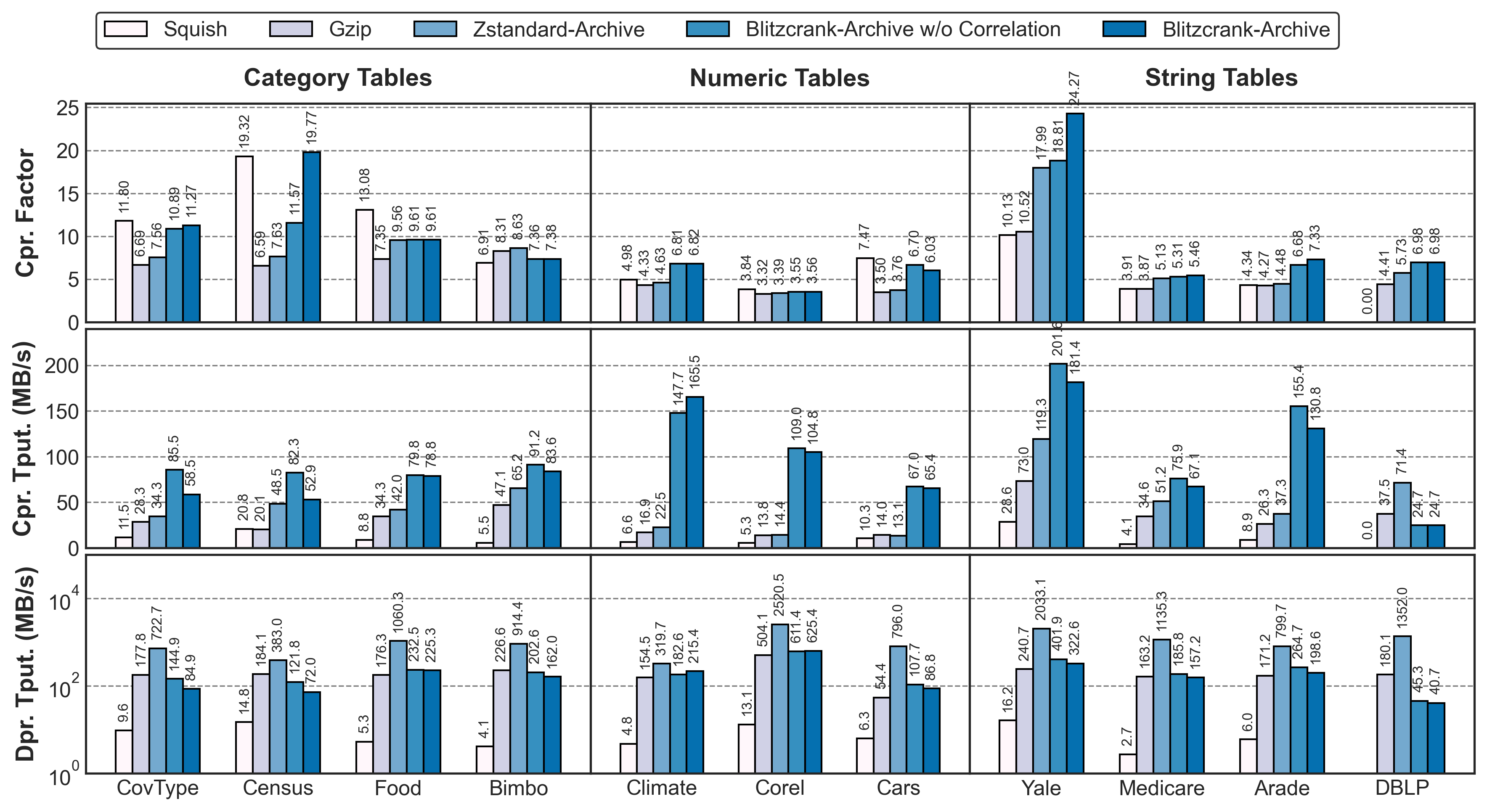}
  \caption{Compression for Archive \textnormal{- We compare the performance of archive-focused compressors, including Squish, Gzip, Zstandard and \work on real tables. Zstandard and \work are in the table archival mode. }}
  \label{compress_archive}
\end{figure*}
We also evaluate \work for the table archive task. Each compressor compresses the given table in memory and then decompresses it to the original, which is typical for general-purpose compressors. We record the compression factor, the compression time, and the decompression time. Archive compression focuses on the compression factor and throughput, rather than the latency. Squish is omitted for compressing \textit{dblp}, because it does not support JSON.

\cref{compress_archive} shows the archive results. For the compression factor, \work compresses categorical and numeric tables comparably to Squish, and outperforms it on string tables. This is because our semantic string model is more effective than the letter-by-letter method. Also, \work-Archive doubles the compression factor in \textit{US Census 1990} compared to \work in random access mode. For the compression speed, \work is $20\times$ faster than Squish on average for both insertion and access latency, thanks to our semantic models and delayed coding. Zstandard excels at archiving, which is expected given its use of a dictionary for compression. In the table archiving task, Zstandard can build a dictionary based on a longer context than in the random access task. The use of longer dictionary entries allows for the output of more bytes per access, thus increasing throughput \cite{DBLP:journals/tit/ZivL78}. 

\subsection{Semi-Structured Data and Time Series}
\begin{table}[t!]
    \centering
        \caption{Compression Factors \textnormal{- \work excels in compressing JSONs; Only \work can benefit from the ARMA model.}}
    \begin{tabular}{cccc}
             & \textbf{Gzip} & \textbf{Zstd} & \textbf{Blitz.} \\
    \midrule
    Relation & 4.13                     & 4.86          & 5.32         \\ 
    JSON     & 4.41                     & 5.73          & 6.98      \\
    $\uparrow$   &  6.78\%              & 17.9\%        & \textbf{31.2\%} \\
    \midrule
    Original      & 4.33                  & 4.64                  & 4.90                \\
    Residual      & 4.29                  & 4.31                  & 6.81                \\
    $\uparrow$    & -0.92\%   & -7.11\%   & \textbf{39.0\%} \\
    \midrule
    \end{tabular}
    \label{cpf_table}
\end{table}
    
\textit{dblp} \cite{dblp} is in JSON format. To investigate how \work performs in different data formats, we extract each attribute from \textit{dblp}, forming new columns. These extracted columns can form a relational table with empty values, called \textit{relation-dblp}. All compressors are in archive mode in this experiment. \cref{cpf_table} shows that the compressors for \textit{relation-dblp} do not reach the high compression factors of \textit{dblp}, but their factors do not reduce proportionally. This is because of redundant repeat property keys in JSON objects. Including these keys enhances the compression factor, with \work being the most efficient compressor, gaining the most from JSON structures.

We then investigate whether the time-series model helps our numeric model improve performance. The time-series model converts the continuous data into residuals through the ARMA model. We compare the compression factor of each compression algorithm on residuals and the original data. We use the \textit{Jena-Climate} in this experiment. \cref{cpf_table} shows the results. Among the compression methods we tried, only \work benefits from the residuals. Residuals typically follow a normal distribution and agree with the assumption of our numeric model. Residuals have fewer outliers and are easier to compress than original data. We estimate the entropy of the time-series data and its residual by discretizing the values into 512 buckets. The results meet our expectations: the entropy is reduced by approximately 30\%, which agrees with the improvement of 39\% shown in \cref{cpf_table}.

% \end{multicols}

\end{document}